\documentclass[a4paper]{article}
\usepackage[figuresright]{rotating}
\usepackage{multirow}
\usepackage{graphics}
\usepackage{graphicx}
\usepackage{makecell}
\usepackage{color}
\usepackage{graphicx}
\usepackage{amsmath}
\usepackage{amsthm}
\usepackage{amssymb}
\usepackage{subfig}
\usepackage{float}
\usepackage{balance}
\usepackage{cite}
\usepackage{booktabs}
\usepackage{cases}
\usepackage{bm}
\usepackage{mathrsfs}
\usepackage[top=2.5cm, bottom=2.5cm, left=3cm, right=3cm]{geometry}

\makeatletter \@addtoreset{equation}{section} \makeatother

\theoremstyle{plain}
\numberwithin{equation}{section}
\newtheorem{theorem}{Theorem}[section]
\newtheorem{lemma}{Lemma}[section]
\newtheorem{Definition}{Definition}[section]
\newtheorem{remark}{Remark}[section]	

\newtheorem{corollary}{Corollary}[section]
\newtheorem{example}{Example}[section]

\begin{document}
    \begin{center}
    \textbf{\LARGE{Several new classes of MDS symbol-pair codes derived from matrix-product codes}}\footnote {
        Xiujing Zheng, Liqi Wang and Shixin Zhu are with School of Mathematics, Hefei University of Technology, Hefei, China
        
        Email:  xiujingzheng99@163.com(X. Zheng);
        liqiwangg@163.com(L. Wang);    
        zhushixin@hfut.edu.cn(S. Zhu)
        }\\
    \end{center}

    \begin{center}
        {Xiujing Zheng \ Liqi Wang \ Shixin Zhu}
    \end{center}

    \begin{center}
    \textit{School of Mathematics, Hefei University of Technology, Hefei, 230009, P. R. China}
    \end{center}

    \begin{abstract}
        In order to correct the pair-errors generated during the transmission of modern high-density data storage that the outputs of the
        channels consist of overlapping pairs of symbols,
        a new coding scheme named symbol-pair code is proposed. 
        The error-correcting capability of the symbol-pair code is determined by its minimum symbol-pair distance. For such codes, the larger the minimum symbol-pair distance, the better. It is a challenging task to construct symbol-pair codes with optimal parameters, especially, maximum-distance-separable (MDS) symbol-pair codes. 
        In this paper, the permutation equivalence codes of matrix-product codes with underlying matrices of orders $3$ and $4$ are used 
        to extend the minimum symbol-pair distance, four new classes of MDS symbol-pair codes and a new class of AMDS symbol-pair codes are derived.

    \end{abstract}

    \section{Introduction}\label{sec1}

    In the traditional information transmission model, we typically divide the information into individual information units to analyze the noisy channel.
    With the increasing demand for data storage, the information sometimes needs to be transmitted in the form of overlapping symbols.
    When the information is transmitted in a channel that outputs pairs of overlapping symbols,
    there are always errors in the writing and reading of symbols.
    To overcome the errors of information transmitted in such overlapping symbols, 
    symbol-pair codes were firstly proposed by Cassuto and Blaum \cite{CB10}. 
    The symbol-pair coding theory further matured by  Cassuto and Blaum \cite{CB11}, where they studied asymptotic bounds on the code rate.
    Shortly after, Cassuto and Litsyn \cite{CL11} presented the Gilbert-Varshamov bounds on code rates for symbol-pair codes. 
    
    In general, a code of length $n$ with size $M$ and minimum symbol-pair distance $d_{sp}(\mathcal{C})$ is called an $(n,M,d_{sp}(\mathcal{C}))$ symbol-pair code.
    Analogous to classical error-correcting codes, the parameters of symbol-pair codes are mutually restricted, 
    and there is also a so-called Singleton bound for symbol-pair codes. 

    \begin{lemma}\cite{CKW12}\label{le1.1}(Singleton bound with respect to symbol-pair codes)
        Let $\mathbb{F}_q$ be the finite field and $2\leq d_{sp}(\mathcal{C})\leq n$. If $\mathcal{C}$ is a symbol-pair code with parameters $(n,M,d_{sp}(\mathcal{C}))_q$,  
        then $M\leq q^{n-d_{sp}(\mathcal{C})+2}$. 
        Particularly, if such bound is attained, then the symbol-pair codes are called  maximum-distance-separable (MDS) symbol-pair codes, abbreviated as MDS $(n,d_{sp}(\mathcal{C}))_q$  symbol-pair codes. 
        If $M= q^{n-d_{sp}(\mathcal{C})+1}$, then the symbol-pair codes are called almost  maximum-distance-separable (AMDS) symbol-pair codes, abbreviated as AMDS $(n,d_{sp}(\mathcal{C}))_q$  symbol-pair codes. 
    \end{lemma}

    The research on symbol-pair codes attracted the interest of many scholars. 
    In \cite{YBS16}, Yaakobi et al. generalized some results for symbol-pair codes to more general $b$-symbol codes. 
    Concomitantly, the Singleton-type bound for $b$-symbol codes was given by Ding et al.\cite{DZG18}. 
    Furthermore, some results on the distributions of symbol-pair weights for linear codes were obtained(see \cite{DWLS19,SZW18,ML20}).
    In \cite{LXY18,HTM14,YBS12,HMH15}, various decoding algorithms for symbol-pair codes were also proposed.

    One of the main tasks in symbol-pair coding theory is to design symbol-pair codes with good parameters.
    Specially, MDS symbol-pair codes are optimal since they have the highest error detection and correction capability for the same code length and code size.
    In \cite{CJKWY13,CKW12}, Chee et al. employed a variety of methods to construct MDS symbol-pair codes, 
    such as Euler diagrams, classical MDS codes, interleaving and extending classical MDS codes.
    Afterwards, numerous scholars have devoted to the construction of MDS symbol-pair codes(see\cite{DGZZZ18,KZL15,ZYST21,LELP23,CLL17} and the relevant references therein). 
    Due to the nice algebraic structure of constacyclic codes, including cyclic codes and negacyclic codes, they have been applied extensively to the construction of symbol-pair codes, especially constacyclic codes with repeated roots. According to these constacyclic codes, MDS symbol-pair codes with various lengths and minimum symbol-pair distances have been obtained. (see \cite{KZL15,KZZLC18,LG17,CLL17,ML22.3,ML22.5,TL23,L23,DNS20}  and the relevant references therein).
    At the same time, some MDS symbol-pair codes over finite rings have also been derived from constacyclic codes(see\cite{DKKSSY19,DKSSGM21,DST21,DNSS18} and the relevant references therein).

    Matrix-product codes over finite fields were first introduced by Blackmore et al. \cite{BN01}, which can be viewed as generalizations
    of Reed-Muller codes and some special constructions of codes, such as  Plotkin's $(u \vert u+v)$-construction, 
    ternary $ (u+v+\omega \vert 2u+v \vert u)$-construction.
    Since then, the research on the properties and applications of matrix-product codes attracted the attentions of a lot of scholars(see \cite{A08,FLL13,LEL23,HLR09} and the relevant references therein).  
    In \cite{HLR09}, Hernando et al. considered the case in which the codes that form the matrix-product codes are nested linear codes and the minimum distances were also computed.
    Recently, Luo et al.\cite{LELP23} proposed a lower bound for the symbol-pair distance of matrix-product codes 
    and some new classes of MDS symbol-pair codes were derived from matrix-product codes.

    \begin{table}
        \caption{Some known MDS symbol-pair codes over finite fields}
        \begin{center}
        	\renewcommand{\arraystretch}{1.3}
            \begin{tabular}{ccccc}
                \toprule
                $q$ &    length $\mathcal{N}$ & $d_{sp}(\mathcal{C})$ &  Refs.  \\
                \midrule
                \multirow{2}{*}{prime power}  & $\mathcal{N}=lp^s$, $l\geq t(q+1)$ coprime to $p$,     & \multirow{2}{*}{4}  & \multirow{2}{*}{\cite{ZYST21}}         \\
                &  $t\vert l$, $\frac{l}{t}\vert (q^2-1)$      &                     &                     \\

                prime power   & $\mathcal{N}\vert (q^2-1)$ and $\mathcal{N}>q+1$  & 5  &     \cite{KZL15}\\
                
                prime power   & $5\leq \mathcal{N}\leq q^2+q+1$                    & 5  &     \cite{DGZZZ18}\\
                prime power   & max $\{6,q+2\}\leq \mathcal{N}\leq q^2$  &  6& \cite{DGZZZ18}\\

                prime power, $q\equiv 1\ ({\rm mod}\ 3)$ & $\frac{q^2+q+1}{3}$  & 6 & \cite{KZL15}\\
                prime power   & $q^2+1$  &  6  & \cite{KZL15}\\

                \multirow{2}{*}{prime power} & $r\vert (q-1) $, $\mathcal{N}\vert (q^3-1)$,              & \multirow{2}{*}{5} & \multirow{2}{*}{\cite{LG17}} \\
                                             & $\mathcal{N}r \nmid(q-1)$, $(\frac{q-1}{r},\mathcal{N})=1$ &                    &  \\
    
                \multirow{2}{*}{prime power} & $\mathcal{N}r\vert (q-1)(q^2+1)$, $\mathcal{N}r\nmid (q^2-1)$,   & \multirow{2}{*}{6}  & \multirow{2}{*}{\cite{LG17}}  \\
                                             &    $(\frac{q-1}{r},\mathcal{N})=1$                     &                     &  \\
                
                \multirow{2}{*}{prime power} & $\mathcal{N}\vert (q^2-1)$, $\mathcal{N}$ odd or        & \multirow{2}{*}{6}& \multirow{2}{*}{\cite{LG17}} \\	
                                             & $\mathcal{N}$ even and $v_2(\mathcal{N})<v_2(q^2-1)$    &                   &    \\

                odd prime power $q\geq 3$  & $\mathcal{N}\geq q+4$, $\mathcal{N}\vert (q^2-1)$ & 6 & \cite{CLL17}\\
    
                odd prime $p\geq 5$ & $\mathcal{N}=lp$, $l>2$, ${\rm gcd}(l,p)=1$, $l\vert(p-1)$ & 5 & \cite{CLL17}\\
                odd prime $p\geq 5$ & $3p$ & 6,\ 7& \cite{CLL17} \\

                odd prime $p$ & $p^2+p$  &  6 & \cite{KZZLC18}   \\
                odd prime $p$ & $2p^2-2p$  &  6 & \cite{KZZLC18}   \\

                odd prime, $p\equiv 1\ ({\rm mod}\ 3)$  & $3p$  & $ 8,10,12$ & \cite{CLL17} \cite{ML22.3} \cite{TL23} \\

                odd prime $p$   & $4p$  &  7 & \cite{KZZLC18}  \cite{ML22.5}  \\

                odd prime, $p\equiv 1\ ({\rm mod}\ 5)$  & $5p$  &  7,\ 8,\ 12 & \cite{L23}  \cite{ML22.5} \\

                even prime power    & $2q+2$   &  7, $2q-1$  & \cite{LELP23} \\    

                \multirow{2}{*}{prime power} & \multirow{2}{*}{$2n$, $n\in[2,q]$}   & $ d_{sp}(\mathcal{C})=2l+1$, &\multirow{2}{*}{\cite{LELP23}}\\
                                             &                                      & $1\leq l\leq n-1$        &                              \\

                prime power       & $3n$, $n\in[3,q]$   &  6 & \cite{LELP23} \\    
                prime power and $q\neq 4,5$,     &  \multirow{2}{*}{$q^2+q$}     &    \multirow{2}{*}{6}   &     \multirow{2}{*}{\cite{LELP23}}  \\
                $q\neq 2^t$ with odd integer $t$ &                               &                         &            \\
                    
                odd prime power, $q\equiv 1\ ({\rm mod}\ 3)$  & $3mp$, $m\in[1,\frac{q}{p}]$   &  7 & \cite{LELP23} \\
                prime power, $q\equiv 1\ ({\rm mod}\ 3)$  & $3mp$, $m\in[1,\frac{q}{p}]$   &  10 & \cite{LELP23} \\

                prime power, $q\equiv 1\ ({\rm mod}\ 3)$  & $3n$, $n\in[4,q]$   &  7 & Theorem  \ref{the3.2} \\
                prime power, $q\equiv 1\ ({\rm mod}\ 3)$  & $3n$, $n\in[4,q]$   &  8 & Theorem  \ref{the3.1}\\
                prime power, $q\equiv 1\ ({\rm mod}\ 3)$  & $3n$, $n\in[5,q]$   &  10 & Theorem \ref{the3.3} \\

                prime power, $q\equiv 1\ ({\rm mod}\ 4)$  & $4n$, $n\in[4,q]$   &  6 & Theorem  \ref{the3.4} \\

                \bottomrule
            \end{tabular}
        \end{center}
        \label{tab1}
    \end{table}

    In Table 1, we summarize some of the known MDS symbol-pair codes over finite fields.
    It can be seen that most MDS symbol-pair codes with minimum symbol-pair distance $d_{sp}(\mathcal{C})>6 $ are restricted to the finite field $\mathbb{F}_p$, where $p$ is a prime number. 
    Going on the line of the study in construting symbol-pair codes from matrix-product codes, five new classes of symbol-pair codes are derived.
    The MDS symbol-pair codes obtained in such paper have more flexible lengths and are over more general fields.
    Particularly, the MDS symbol-pair codes constructed in Theorems \ref{the3.2} and \ref{the3.3} generalize the lengths of two classes of MDS symbol-pair codes obtained in \cite{LELP23}.

    The rest of the paper is organized as follows: 
    In Sect.\ref{sec2}, some  notations and basic results about linear codes and symbol-pair codes over finite fields are reviewed. 
    In Sect.\ref{sec3}, several new classes of MDS symbol-pair codes and AMDS symbol-pair codes are derived from matrix-product codes. 
    We conclude this paper in Section \ref{sec4}.

    \section{Preliminaries}\label{sec2}

	Throughout this paper, let $q$ be a power of a prime $p$, $\mathbb{F}_{q}$ be a finite field with $q$ elements and $\mathbb{F}^*_{q}=\mathbb{F}_{q}\setminus \{0\}$. 
	A $q$-ary linear code $\mathcal{C}$ of length $n$ with dimension $k$, denoted by $[n,k]_{q}$, is a $k$-dimensional subspace of $\mathbb{F}^n_{q}$.
	The Hamming distance of any two codewords $\textbf{x}=(x_0,x_1,\cdots,x_{n-1})$, $\textbf{y}=(y_0,y_1,\cdots,y_{n-1})\in \mathcal{C}$ is $d_H(\textbf{x},\textbf{y})=wt(\textbf{x}-\textbf{y})$,
    where $wt(\textbf{x}-\textbf{y})$, called the Hamming weight of $\textbf{x}-\textbf{y}$, denotes the number of nonzero components of $\textbf{x}-\textbf{y}$.
	The minimum Hamming distance $d_H$ of $\mathcal{C}$ is the minimum Hamming distance between any two distinct codewords of $\mathcal{C}$.
	An $[n,k,d_H]_{q}$ linear code is an $[n,k]_{q}$ linear code with minimum  Hamming distance $d_H$. 
    If $d_H=n-k+1$, then the linear code $\mathcal{C}$ is called an MDS code. The dual of a linear code $\mathcal{C}$ is defined by the set
	$$\mathcal{C}^{\bot}=\{\textbf{z}\in \mathbb{F}_{q}^n : \sum_{i = 0}^{n-1}x_iz_i=0,for \ all \ \textbf{x}\in \mathcal{C}\}.$$

     Generalized Reed-Solomon (GRS) code, one of the best known MDS code families, is defined as follows.
	Suppose that $\textbf{a}=(\alpha_0,\alpha_1, \cdots, \alpha_{n-1}) \in \mathbb{F}_{q}^n$ and $\textbf{v}=(v_0, v_1,\cdots, v_{n-1}) \in (\mathbb{F}^*_{q})^n$,
	where $\alpha_0,\alpha_1, \cdots,$ $\alpha_{n-1}$ are $n$ distinct elements of $\mathbb{F}_{q}$ and
	$v_0, v_1, \cdots, v_{n-1}$ are $n$ nonzero elements of  $\mathbb{F}_{q}$ ($v_i$ can be the same).
	For an integer $k$ with $1\leq k\leq n$, let 
	$$\mathbb{F}_{q}[x]_k=\{f(x)\in \mathbb{F}_{q}[x] : {\rm deg}(f(x))\leq k-1 \}.$$
	Then a GRS code of length $n$ associated with $\textbf{a}$ and $\textbf{v}$ is defined as 
	$$GRS_k(\textbf{a},\textbf{v})=\{v_0f(\alpha_0), v_1f(\alpha_1), \cdots, v_{n-1}f(\alpha_{n-1}) : for \ all \ f(x)\in \mathbb{F}_{q}[x]_k\}.$$
	The elements $\alpha_0, \alpha_1, \cdots, \alpha_{n-1}$ are called the evaluation points of $GRS_k(\textbf{a},\textbf{v})$, 
	and $v_0, v_1, \cdots, v_{n-1}$ are called the column multipliers of $GRS_k(\textbf{a},\textbf{v})$.

	It is well known that a GRS code $GRS_k(\textbf{a},\textbf{v})$ is an MDS code with parameters $[n,k,n-k+1]_{q}$ and its generator matrix  is 
	
	$$ G_k  =  \begin{pmatrix}
		v_0              & v_1                   & \cdots & v_{n-1}             \\
		v_0\alpha_0     & v_1\alpha_1           & \cdots & v_{n-1}\alpha_{n-1}                 \\
		\vdots           & \vdots                & \ddots & \vdots                 \\
		v_0\alpha_0^{k-1} & v_1\alpha_1^{k-1}   & \cdots & v_{n-1}\alpha_{n-1}^{k-1} \\
	
	\end{pmatrix}.$$

    As we know, the dual of a GRS code $GRS_k(\textbf{a},\textbf{v})$ is still a GRS code $GRS_{n-k}(\textbf{a},\textbf{v}^{\prime})$ 
    for a vector $\textbf{v}^{\prime}=(v_0^{\prime}, v_1^{\prime}, \cdots, v_{n-1}^{\prime})\in (\mathbb{F}^*_{q})^{n}$.
    Additionally, the generator matrix of $GRS_{n-k}(\textbf{a},\textbf{v}^{\prime})$ is the parity-check matrix  of $GRS_k(\textbf{a},\textbf{v})$.
    Then the parity-check matrix of $GRS_k(\textbf{a},\textbf{v})$  has the following form

	$$ H_{n-k}  =  \begin{pmatrix}
		v_0^{\prime}              & v_1^{\prime}                   & \cdots & v_{n-1}^{\prime}             \\
		v_0^{\prime}\alpha_0     & v_1^{\prime}\alpha_1           & \cdots & v_{n-1}^{\prime}\alpha_{n-1}                 \\
		\vdots           & \vdots                & \ddots & \vdots                 \\
		v_0^{\prime}\alpha_0^{n-k-1} & v_1^{\prime}\alpha_1^{n-k-1}  & \cdots & v_{n-1}^{\prime}\alpha_{n-1}^{n-k-1} \\
	
	\end{pmatrix}.$$

    \subsection{Matrix-product codes}
    Assume that $A=(a_{ij})$ is an $M \times N$ matrix with entries in $\mathbb{F}_{q}$,
    and $\mathcal{C}_1, \mathcal{C}_2, \cdots, \mathcal{C}_{M}$ is a family of codes of length $n$ over $\mathbb{F}_{q}$.
    The matrix-product (MP) code $\mathcal{C}=[\mathcal{C}_1,\mathcal{C}_2,\cdots, \mathcal{C}_M]\cdot A$ of length $nN$ is the set of all matrix products as follows:
    $$\{(\textbf{c}_1,\textbf{c}_2,\cdots,\textbf{c}_M)\cdot A : \textbf{c}_i\in \mathcal{C}_i,\  i=1,2,\cdots, M  \} $$
    $$=\{(\sum_{m=1}^{M}a_{m1}\textbf{c}_m,\sum_{m=1}^{M}a_{m2}\textbf{c}_m,\cdots,\sum_{m=1}^{M}a_{mN}\textbf{c}_m): \textbf{c}_i\in \mathcal{C}_i,\  i=1,2,\cdots, M \}.$$
   
    If $\mathcal{C}_i$ is a linear code with generator matrix $G_i$, where $i=1,2,\cdots, M$, 
    then the generator matrix $G$ of the matrix-product code $\mathcal{C}$ is 

	$$ G  =  \begin{pmatrix}
		a_{11}G_1    &   a_{12}G_1     & \cdots & a_{1N}G_1     \\
		a_{21}G_2    &   a_{22}G_2     & \cdots & a_{2N}G_2     \\
		\vdots       &  \vdots         & \ddots & \vdots        \\
		a_{M1}G_M    &  a_{M2}G_M      & \cdots & a_{MN}G_M    \\
	
	\end{pmatrix}.$$

    Recall that a square matrix $A$ is non-singular if there exists a square matrix $B$ such that  $A\cdot B=B\cdot A=I$. Let $A_t$ be the matrix consisting of the first $t$ rows of $A=(a_{ij})_{M\times N}$. For $1 \leq j_1\leq j_2\leq \cdots \leq j_t\leq N $, 
    let $A_t(j_1,j_2,\cdots,j_t)$ be a $t\times t$ matrix consisting of columns $j_1,j_2,\cdots,j_t$ of $A_t$. A matrix that is non-singular by columns (NSC) is defined as follows.
    
    \begin{Definition}\cite{BN01}
        A matrix $A=(a_{ij})_{M\times N}$ is called {\bf non-singular by columns (NSC)} 
        if for any $1\leq t\leq M$ and $1\leq j_1,j_2,\cdots,j_t\leq N$, the square matrix $A_t(j_1,j_2,\cdots,j_t)$ is non-singular.       
    \end{Definition}

    \begin{remark}
        It was also shown in \cite{BN01} that for $M\geq 2$,  there exists an NSC $M\times N$ matrix over $\mathbb{F}_q$ if and only if $M\leq N\leq q$.
        Thus, we will default to $M\leq N$ as we consider the NSC matrix.
    \end{remark}

    For each integer $1\leq t\leq M$, denote the linear code of length $N$ generated by the matrix $A_t$ as $\mathcal{A}_t$.
    The following lemma gives a relationship between the parameters of a matrix-product code and the parameters of the codes that make it up.

    \begin{lemma}\cite{BN01}\label{le2.1}
        Let $\mathcal{C}=[\mathcal{C}_1,\mathcal{C}_2,\cdots, \mathcal{C}_M]\cdot A$ be an MP code and $\mathcal{C}_i$ be an $[n,k_i,d_H(\mathcal{C}_i)]_q$ linear code,
        where $1\leq i\leq M$ and $rank(A)=M$. Then $\mathcal{C}$ is an $[nN,\sum_{i = 1}^{M} k_i,d_H(\mathcal{C}) ]_q$ code, where
        $$d_H(\mathcal{C})\geq min \{d_H(\mathcal{C}_1)\cdot d_H(\mathcal{A}_1), d_H(\mathcal{C}_2)\cdot d_H(\mathcal{A}_2), \cdots, d_H(\mathcal{C}_M)\cdot d_H(\mathcal{A}_M)\}.$$
        Specially, if the matrix $A$ is also an NSC matrix, then $d_H(\mathcal{A}_i)=M-i+1$ for $1 \leq i\leq M$.
    \end{lemma}

    If $A$ is a square matrix, then for the dual of the MP code, we have the following result. 
    \begin{lemma}\cite{BN01}\label{le2.2}
        Suppose that $\mathcal{C}=[\mathcal{C}_1,\mathcal{C}_2,\cdots, \mathcal{C}_M]\cdot A$ is an MP code and $A$ is an $M\times M$ non-singular matrix.
        Then $\mathcal{C}^{\bot }$ is also an MP code and is given by 
        $$([\mathcal{C}_1,\mathcal{C}_2,\cdots, \mathcal{C}_M]\cdot A)^{\bot }=[\mathcal{C}_1^{\bot },\mathcal{C}_2^{\bot },\cdots, \mathcal{C}_M^{\bot }]\cdot (A^{-1})^T.$$
    \end{lemma}

    Surely, from Lemma \ref{le2.2}, the following corollary holds.

    \begin{corollary}\label{co3.1}
    Let $A^{-1}=(b_{ij})_{M \times M}$ and $H_i$ be the parity-check matrix of $\mathcal{C}_i$, where $i=1,2,\cdots,M$. 
    Then the parity-check matrix of the MP code $\mathcal{C}$ is

	$$ H  =  \begin{pmatrix}
		b_{11}H_1    &   b_{21}H_1     & \cdots & b_{M1}H_1     \\
		b_{12}H_2    &   b_{22}H_2     & \cdots & b_{M2}H_2     \\
		\vdots       &  \vdots         & \ddots & \vdots        \\
		b_{1M}H_M    &  b_{2M}H_M      & \cdots & b_{MM}H_M    \\
	\end{pmatrix}.$$ 

    \end{corollary}

    The following lemma is very important for the proofs in the sequel.

    \begin{lemma}\cite{LELP23}\label{le2.3}
        Let $\mathcal{C}_M\subseteq \cdots \subseteq\mathcal{C}_1$ be nested linear codes of length $n$ over $\mathbb{F}_q$ and $A$ be an $M\times N$ NSC matrix.
        Suppose that $\textbf{c}$ is a codeword of the corresponding MP code $\mathcal{C}=[\mathcal{C}_1,\mathcal{C}_2,\cdots, \mathcal{C}_M]\cdot A$.
        Then $\textbf{c}$ can be written as $\textbf{c}=(\textbf{c}_1,\textbf{c}_2,\cdots,\textbf{c}_N)$, where $\textbf{c}_i$ is a vector of length $n$, $i=1,2\cdots, N$,
        and if for $1\leq k\leq M-1$, there are precisely $k$ of $\textbf{c}_1,\textbf{c}_2,\cdots,\textbf{c}_N$ that are zero vectors, then for any $i=1,2\cdots, N$, $\textbf{c}_i\in \mathcal{C}_{k+1}$.
        If the number of zero vectors among $\textbf{c}_1,\textbf{c}_2,\cdots,\textbf{c}_N$ is greater than $M-1$, then $\textbf{c}=\textbf{0}$.

    \end{lemma}

    \subsection{Symbol-pair codes}

    In this subsection, we will review some basic notions and properties of symbol-pair codes.
    
    Let $\Omega  $ denote an alphabet consists of $q$ elements, and we call the elements in $\Omega$ as symbols.  
    For any vector $\textbf{u}=(u_0,u_1,\cdots,u_{n-1})\in \Omega^n$, the symbol-pair read vector of $\textbf{u}$ is defined as follows.
    $$\varTheta  (\textbf{u})=((u_0,u_1),(u_1,u_2),\cdots ,(u_{n-2},u_{n-1}),(u_{n-1},u_0)).$$
    Clearly, for any vector $\textbf{u}\in \Omega^n$, there exists a unique symbol-pair read vector $\varTheta  (\textbf{u})$ in $(\Omega\times \Omega)^n $.
    In this paper, we assume that $\Omega=\mathbb{F}_{q}$. 

    Let $\textbf{u}=(u_0,u_1,\cdots,u_{n-1})$, $\textbf{v}=(v_0,v_1,\cdots,v_{n-1})$ be any two vectors in $\mathbb{F}_{q}^n$, 
    the symbol-pair distance by using Hamming distance from $\textbf{u}$ to $\textbf{v}$ is  
    $$d_{sp}(\textbf{u},\textbf{v})=d_H(\varTheta (\textbf{u}),\varTheta (\textbf{v}))=| \{0\leq i\leq n-1:(u_i,u_{i+1})\neq (v_i,v_{i+1}) \}| ,$$
    where the subscripts are reduced modulo $n$. The  minimum symbol-pair distance of a symbol-pair code is defined as
    $$d_{sp}(\mathcal{C})=min \{ d_{sp}(\textbf{u},\textbf{v}): \textbf{u},\textbf{v}\in \mathcal{C},\textbf{u}\neq \textbf{v}\}.$$
    An $(n,M,d_{sp}(\mathcal{C}))$ symbol-pair code $\varTheta (\bf{\mathcal{C}}) $ is a subset $\mathcal{C}\subset \mathbb{F}_{q}^n$ of length $n$ with size $M$ and minimum
    symbol-pair distance $d_{sp}(\mathcal{C})$, where $M=|\mathcal{C} | $.
    For any vector $\textbf{u}\in \mathbb{F}_{q}^n$, we define the symbol-pair weight of $\textbf{u}$ as 
    $$w_{sp}(\textbf{u})=w_H(\Theta(\textbf{u}))=| \{0\leq i\leq n-1:(u_i,u_{i+1})\neq (0,0) \}|.$$
    Specially, if $\mathcal{C}$ is a linear code, then we can get
    $$d_{sp}(\mathcal{C})=min \{ w_{sp}(\textbf{u}): \textbf{u}\in \mathcal{C}\}.$$
   
    Assume that  $\textbf{c}=(c_0,c_1,\cdots, c_{n-1})$ is a codeword of length $n$ in $\mathcal{C}$, then the symbol-pair read vector of $\textbf{c}$ is 
    $$\Theta(\textbf{c})=\{(c_i,c_{i+1}): 0\leq i\leq n-1\}.$$ 

    Define two subsets from $\Theta(\textbf{c})$,
    $$\Theta_1(\textbf{c})=\{({c_i},c_{i+1})\in \Theta(\textbf{c}): c_i\neq 0,\ 0\leq i\leq n-1 \},$$
    and 
    $$\Theta_2(\textbf{c})=\{({c_i},c_{i+1})\in \Theta(\textbf{c}): c_i\neq 0,c_{i+1}= 0,\ 0\leq i\leq n-1 \}.$$
    From the definitions of Hamming weight and symbol-pair weight of $\textbf{c}$, one can get 
    $w_H(\textbf{c})=|\Theta_1(\textbf{c})| $
    and
    $$w_{sp}(\textbf{c})=w_H(\textbf{c})+|\Theta_2(\textbf{c})|=|\Theta_1(\textbf{c})|+|\Theta_2(\textbf{c})|.$$ 
    Denote $I=|\Theta_2(\textbf{c})|,$ then $I=w_{sp} (\textbf{c})-w_H(\textbf{c})\leq n-w_H(\textbf{c})$. If $0<d_H(\mathcal{C})<n$, combining with $1\leq I\leq w_H(\textbf{c})$, we can get 
    $$w_H(\textbf{c})+1\leq w_{sp}(\textbf{c})\leq min\{2w_H(\textbf{c}),n\}.$$
    Additionally, there exists a connection  between the minimum Hamming distance and the minimum symbol-pair distance which was proved in \cite{CB11}. When $0<d_H(\mathcal{C})<n$,
    $$d_H(\mathcal{C})+1\leq d_{sp}(\mathcal{C})\leq min\{2d_H(\mathcal{C}),n \}.$$
    Particularly, if $d_H(\mathcal{C})=0$ or $n$, then we can easily get $ d_{sp}(\mathcal{C})=d_H(\mathcal{C})$.

    The concept of code equivalence is very important in coding theory. As we know, the equivalent codes have the same parameters.
    Let $\mathcal{C}_1$ and $\mathcal{C}_2$ be two linear codes over $\mathbb{F}_q$, 
    then they are said to be equivalent if $\mathcal{C}_1$ can be obtained from $\mathcal{C}_2$ by any combination of the following transformations.
    (1): The permutation of the code coordinates.
    (2): Multiplication of elements in a fixed position by a non-zero scalar in $\mathbb{F}_q$.
    (3): A field automorphism $\tau : \mathbb{F}_q\rightarrow \mathbb{F}_q$ to each component of the code.

    If $\mathcal{C}_1$ is  obtained from $\mathcal{C}_2$ only by (1), then $\mathcal{C}_1$ and $\mathcal{C}_2$ are called permutation equivalent.
    Permutation equivalent codes of course have the same minimum Hamming distance, however, they do not retain the symbol-pair distance.
    Therefore, we can find a code that permutates equivalent to a certain code to expand its symbol-pair distance.

\section{New symbol-pair codes from MP codes}\label{sec3}
In this section, assume that $q$ is a prime power, we will construct several new classes of symbol-pair codes from the codes that  permutate equivalent to the MP codes.
To better state our proof, we need the following definition.

\begin{Definition}
    Let $\textbf{c}=(c_0,c_1,\cdots,c_{n-1})$ be a vector of length $n$. Then the support of $\textbf{c}$  is defined by
    $$supp(\textbf{c})=\{0\leq i\leq n-1:c_i\neq 0  \} .$$
    
\end{Definition}

Denoting the number of elements in $supp(\textbf{c})$ as $S$, i.e., $S=|supp(\textbf{c})|$.
Specially, if $\textbf{c}$ is a codeword of the code $\mathcal{C}$, then $S$ is exactly the Hamming weight of $\textbf{c}$.
For example, assume that $\textbf{c}=(1,0,1,0,0,1,1,0)\in \mathcal{C}$, then $supp(\textbf{c})=\{0,2,5,6\},$ and $S=4.$

Let $\textbf{a}=(\alpha_0,\alpha_1,\cdots,\alpha_{n-1})$ 
and $\textbf{v}=(1,1,\cdots,1)$, where $\alpha_0,\alpha_1,\cdots,\alpha_{n-1}$ are $n$ distinct elements of $\mathbb{F}_q$.
Then $GRS_{i}$ is defined as the GRS code with parameters $[n,n-i,i+1]_q$ whose parity-check matrix is 
$$ H_i=\begin{pmatrix}
    1   & 1   & \cdots & 1 \\
    \alpha_0        & \alpha_1       & \cdots & \alpha_{n-1} \\
    \alpha^2_0        & \alpha^2_1       & \cdots & \alpha^2_{n-1} \\
    \vdots          & \vdots         & \ddots & \vdots    \\
    \alpha^{i-1}_0 & \alpha^{i-1}_1  & \cdots & \alpha^{i-1}_{n-1} \\
\end{pmatrix}_{(i\times n)}.$$
It is easy to see that such GRS codes are nested, namely, $GRS_{n-1}\subseteq GRS_{n-2} \subseteq \cdots \subseteq GRS_1$.

In the following, we use the shorthand notation $[a, b]:=\{a,a+1,\cdots,b\}$ for integers $a<b$.

\subsection{MP codes with the square matrix $A$ of order $3$}
In this subsection, we will construct three classes of MDS symbol-pair codes of length $\mathcal{N}=3n$ from MP codes with the square matrix $A$ of order $3$. 
Let $3\vert(q-1)$, then there must exist a primitive $3$-th root of unity $\omega$ in $\mathbb{F}_q$.
Suppose that $A$ is a $3\times 3$ NSC matrix with the following form:   
$$ A=\begin{pmatrix}
    1   & 1   & 1 \\
    1   & \omega   & \omega^2 \\
    1   & \omega^2    & \omega  \\
\end{pmatrix}.$$

Obviously, we can get 
$$ (A^{-1})^T=\frac{1}{3}\begin{pmatrix}
    1   & 1   & 1 \\
    1   & \omega^2   & \omega\\
    1   & \omega   & \omega^2  \\
\end{pmatrix}.$$ 

We give the following permutations for codewords, which are useful for our constructions.

Let $\textbf{c}=(c_0,c_1,\cdots,c_{3n-1})$ be a codeword of $\mathcal{C}$ of length $3n$, whose coordinates is indexed by the set $[0,3n-1]$.
For each $l\in [0,3n-1]$, we write $l=in+j$, where $i=0,1,2$, $j=0,1,2,\cdots, n-1$.
Then each entry of the vector $\textbf{c}$ can be represented as $c_{i,j}$.

Define a permutation $\rho$ as $\rho(in+j)=i+3j$ and a permutation $\phi$ as 

\begin{equation}
    \phi(i+3j)=\left\{ 
    \begin{array}{l} \notag 
     i+3j,\ {\rm if}\ i=0,2,  \\ 
    \\
     i+3(j+1),\ {\rm if}\ i=1.\\
    \end{array}
    \right.
\end{equation}

Namely,

\begin{equation}\label{eq3.1}
    \begin{matrix}
        c_{0,0},c_{0,1},c_{0,2},c_{0,3},c_{0,4},c_{0,5},\cdots,c_{0,n-1},c_{1,0},\cdots,c_{1,n-1},c_{2,0},\cdots,c_{2,n-1}\\
\downarrow\ \rho\\
c_{0,0},c_{1,0},c_{2,0},c_{0,1},c_{1,1},c_{2,1},c_{0,2},c_{1,2},c_{2,2},\cdots,\cdots, c_{0,n-1},c_{1,n-1},c_{2,n-1}\\
\downarrow\ \phi\\
c_{0,0},c_{1,1},c_{2,0},c_{0,1},c_{1,2},c_{2,1},c_{0,2},c_{1,3},c_{2,2},\cdots,\cdots, c_{0,n-1},c_{1,0},c_{2,n-1}.\\
    \end{matrix}
\end{equation}

\subsubsection{Symbol-pair distance $d_{sp}(\mathcal{C})=8$}

Assume that $\mathcal{C}_1$ is the GRS code $GRS_{1}$ with parameters $[n,n-1,2]_q$ whose parity-check matrix is 
$$ H_1=\begin{pmatrix}
    1   & 1   & \cdots & 1 \\
\end{pmatrix},$$
 $\mathcal{C}_2$ is the GRS code $GRS_{2}$ with parameters $[n,n-2,3]_q$ whose parity-check matrix is 
$$ H_2=\begin{pmatrix}
    1   & 1   & \cdots & 1 \\
    \alpha_0    & \alpha_1   & \cdots & \alpha_{n-1} \\
\end{pmatrix},$$
and $\mathcal{C}_3$ is the GRS code $GRS_{3}$ with parameters $[n,n-3,4]_q$ whose  parity-check matrix is 

$$ H_3=\begin{pmatrix}
    1   & 1   & \cdots & 1 \\
    \alpha_0    & \alpha_1   & \cdots & \alpha_{n-1} \\
    \alpha^2_0    & \alpha^2_1   & \cdots & \alpha^2_{n-1} \\
\end{pmatrix}.$$

Define the MP code

\begin{align}\label{eq3.2}
    \mathcal{C}=[\mathcal{C}_1,\mathcal{C}_2,\mathcal{C}_3]\cdot A.
\end{align}

Obviously, according to Lemma \ref{le2.1}, $\mathcal{C}$ is a $[3n,3n-6]_q$ code, and from Corollary \ref{co3.1}, the parity-check matrix of $\mathcal{C}$ is shown below: 

  \[\begin{split} H &=\frac{1}{3}\begin{pmatrix}
    H_1   & H_1   & H_1 \\
    H_2    & \omega^2H_2   & \omega H_2 \\
    H_3   & \omega H_3  & \omega^2H_3 \\   
\end{pmatrix},\\
&=\frac{1}{3}\begin{pmatrix}  
    1             & \cdots        & 1            & 1                & \cdots     & 1                       & 1        & \cdots     & 1    \\
    1             & \cdots       & 1             &  \omega^2        & \cdots     & \omega^2                 & \omega   & \cdots    & \omega \\
    \alpha_0      & \cdots       & \alpha_{n-1}  & \omega^2\alpha_0 & \cdots     & \omega^2\alpha_{n-1}     & \omega\alpha_0 & \cdots    & \omega\alpha_{n-1}    \\                                      
    1             & \cdots       & 1            &    \omega         & \cdots     & \omega                  &    \omega^2        & \cdots    & \omega^2               \\                                             
    \alpha_0      & \cdots       & \alpha_{n-1}    &  \omega\alpha_0   & \cdots   & \omega\alpha_{n-1}        &  \omega^2\alpha_0    & \cdots & \omega^2\alpha_{n-1}     \\
    \alpha^2_0    & \cdots        & \alpha^2_{n-1}  &   \omega\alpha^2_0 & \cdots &  \omega\alpha^2_{n-1}    &   \omega^2\alpha^2_0    & \cdots   &  \omega^2\alpha^2_{n-1} \\
\end{pmatrix}.
\end{split} \]

The following lemma determines the support for codewords in $\mathcal{C}$ whose Hamming weights do not exceed $6$.

\begin{lemma}\label{le3.1}
    Let $\mathcal{C}$ be the MP code as defined in (\ref{eq3.2}), and $\textbf{c}$ be a codeword of $\mathcal{C}$ with coordinates indexed by the set $[0,3n-1]$.
    Then the following results hold.
    
    (i) There are no codewords in  $\mathcal{C}$ with Hamming weight less than $4$.

    (ii) If $w_H(\textbf{c})=4$, then the support  of $\textbf{c}$ must satisfy 
    $\{(i_1,i_2,i_3,i_4):  jn\leq i_1<i_2<i_3<i_4<(j+1)n,\ j\in [0,2] \} $.
   
    (iii) If $w_H(\textbf{c})=5$, then the support of $\textbf{c}$ must satisfy 
    $\{(i_1,i_2,i_3,i_4,i_5):  jn\leq i_1<i_2<i_3<i_4<i_5<(j+1)n,\ j\in [0,2] \} $.
    
    (iv) If $w_H(\textbf{c})=6$, then the support of $\textbf{c}$ must satisfy  $\{(i_1,i_2,i_3,i_4,i_5,i_6): jn\leq  i_1<i_2<i_3<i_4<i_5<i_6 <(j+1)n, \ j\in [0,2] \} $   
  or $\{(i_1,i_2,i_3,i_4,i_5,i_6): j_1n\leq  i_1<i_2<i_3<(j_1+1)n, \ j_2n \leq i_4<i_5<i_6< (j_2+1)n,\ j_1\neq j_2\ and \ j_1,j_2\in [0,2]  \} $  
    or $\{(i_1,i_2,i_3,i_4,i_5,i_6): j_1n\leq  i_1<i_2<(j_1+1)n, \ j_2n \leq i_3<i_4< (j_2+1)n, \ j_3n \leq i_5<i_6< (j_3+1)n,\  j_1\neq j_2\neq j_3\ and\  j_1,j_2,j_3\in [0,2]\}$ 
    with $\alpha_{i_1}+\alpha_{i_2}=\alpha_{i_3}+\alpha_{i_4}=\alpha_{i_5}+\alpha_{i_6}$, and the subscripts are reduced modulo $n$.  
\end{lemma}

\begin{proof}
    The proof is based on the parity-check matrix of the MP code defined in (\ref{eq3.2}) and Lemma \ref{le2.3}.
    Writing the codeword $\textbf{c}$ of $\mathcal{C}$ as $\textbf{c}=(\textbf{c}_1,\textbf{c}_2,\textbf{c}_3)$, 
    where $\textbf{c}_i$ is a vector of length $n$. 

    (i) If $\textbf{c}_1$, $\textbf{c}_2$ and  $\textbf{c}_3$ are all nonzero  vectors, then we can get  $\textbf{c}_i\in \mathcal{C}_1$. 
    As $\mathcal{C}_1$ is an $[n,n-1,2]_q$ code, then $w_H(\textbf{c})\geq 6$.

    If one of $\textbf{c}_1$, $\textbf{c}_2$, $\textbf{c}_3$ is a zero vector, then we can get  $\textbf{c}_i\in \mathcal{C}_2$. 
    As $\mathcal{C}_2$ is an $[n,n-2,3]_q$ code, then $w_H(\textbf{c})\geq 6$. 

    If one of $\textbf{c}_1$, $\textbf{c}_2$, $\textbf{c}_3$ is a nonzero vector, then we can get  $\textbf{c}_i\in \mathcal{C}_3$. 
    As $\mathcal{C}_3$ is an $[n,n-3,4]_q$ code, then $w_H(\textbf{c})\geq 4$. 
    
    Hence, one can get the Hamming weight of \textbf{c} is always greater than or equal to $4$. (i) holds.

    (ii) For $w_H(\textbf{c})=4$, only one of $\textbf{c}_1$, $\textbf{c}_2$, $\textbf{c}_3$ is a nonzero vector. 
    Since for any $\{(i_1,i_2,i_3,i_4):  jn\leq i_1<i_2<i_3<i_4<(j+1)n,\ j\in [0,2] \}$, from the form of the parity-check matrix,
    the corresponding column vectors of the parity-check matrix are linearly dependent. Hence, the corresponding codewords exist. (ii) holds.

    (iii) For $w_H(\textbf{c})=5$, only one of $\textbf{c}_1$, $\textbf{c}_2$, $\textbf{c}_3$ is a nonzero vector. 
    Since for any $\{(i_1,i_2,i_3,i_4,i_5):  jn\leq i_1<i_2<i_3<i_4<i_5<(j+1)n,\ j\in [0,2] \} $,
    the corresponding column vectors of the parity-check matrix are linearly dependent, the corresponding codewords also exist. (iii) holds.

    (iv) $w_H(\textbf{c})=6$:
    
    If one of $\textbf{c}_1$, $\textbf{c}_2$, $\textbf{c}_3$ is a nonzero vector, 
    then the nonzero vector $\textbf{c}_i$ must satisfy $|supp(\textbf{c}_i)|=6$.
    Since for $\{(i_1,i_2,i_3,i_4,i_5,i_6): jn\leq  i_1<i_2<i_3<i_4<i_5<i_6 <(j+1)n, \ j\in [0,2] \}$, 
    the corresponding column vectors of the parity-check matrix are linearly dependent, the corresponding codewords exist.

    If one of $\textbf{c}_1$, $\textbf{c}_2$, $\textbf{c}_3$ is a zero vector, 
    then the nonzero vector $\textbf{c}_i$ must satisfy $|supp(\textbf{c}_i)|=3 $.
    Since for 
    $\{(i_1,i_2,i_3,i_4,i_5,i_6): j_1n\leq  i_1<i_2<i_3<(j_1+1)n, \ j_2n \leq i_4<i_5<i_6< (j_2+1)n,\  j_1\neq j_2\ and \ j_1,j_2\in [0,2] \} $,  
    the corresponding column vectors of the parity-check matrix are linearly dependent, the corresponding codewords also exist.
    
    If $\textbf{c}_1$, $\textbf{c}_2$ and  $\textbf{c}_3$ are all nonzero vectors,
    then the nonzero vector $\textbf{c}_i$ must satisfy $|supp(\textbf{c}_i)|=2 $.
    Since for  
    $\{(i_1,i_2,i_3,i_4,i_5,i_6): j_1n\leq  i_1<i_2<(j_1+1)n, \ j_2n \leq i_3<i_4< (j_2+1)n, \ j_3n \leq i_5<i_6< (j_3+1)n,\  j_1\neq j_2\neq j_3\ and\  j_1,j_2,j_3\in [0,2]\}$,
    the corresponding column vectors of the parity-check matrix are linearly dependent if and only if  $\alpha_{i_1}+\alpha_{i_2}=\alpha_{i_3}+\alpha_{i_4}=\alpha_{i_5}+\alpha_{i_6}$, where the subscripts are reduced modulo $n$. 
    Hence, (iv) holds.

\end{proof}

With the help of the above lemma, the following theorem holds.

\begin{theorem}\label{the3.1}
    Suppose that $q$ is a power of a prime number $p$ with $q\equiv 1\ ({\rm mod}\ 3)$, then for each $n\in [4,q]$, there exists an MDS $(3n,8)_q$ symbol-pair code.
\end{theorem}

\begin{proof}
    Since in terms of cosets, we can write $\mathbb{F}_q \triangleq \bigcup_{i=0}^{\frac{q}{p}-1}(\chi_i+\mathbb{F}_p)$, where $\chi_0=0$. Let 
    $$\mathcal{F}=(0,1,2,\cdots,p-1,\chi_1,\chi_1+1,\cdots,\chi_1+p-1,\cdots \cdots,\chi_{\frac{q}{p}-1},\cdots,\chi_{\frac{q}{p}-1}+p-1 ),$$
    which consists of all the distinct elements in $\mathbb{F}_q$, and let $\textbf{a}=(\alpha_0,\alpha_1,\cdots,\alpha_{n-1})$ be a vector formed by the first $n$ elements of $\mathcal{F}$.
    Suppose that $\mathcal{C}$ is an MP code as defined in (\ref{eq3.2}) and $\mathcal{D}$ is a code that permutates equivalent to the code $\mathcal{C}$ under 
    the specific permutations $\rho$ and $\phi$ (see (\ref{eq3.1})). Namely,
    \begin{equation*}
        \begin{aligned}
            \mathcal{D}: =&\phi(\rho(\mathcal{C}))  \\
                       : =&\{\phi(\rho(\textbf{c})), \  \forall \textbf{c}\in \mathcal{C}\}. \\
        \end{aligned}
    \end{equation*}
    The codes $\mathcal{C}$ and $\mathcal{D}$ have the same parameters $[3n,3n-6,4]_q$ because they are permutation equivalent.
    We will illustrate that $d_{sp}(\mathcal{D})= 8$ with the help of the support distribution of the codewords $\textbf{c}$ of $\mathcal{C}$.
   
    From Lemma \ref{le3.1}, there are no codewords with Hamming weight less than $4$. 
    Combining with $w_{sp}(\textbf{c})\geq w_H(\textbf{c})+1$ for $w_H(\textbf{c})<3n$, and $w_{sp}(\textbf{c})=w_H(\textbf{c})$ for $w_H(\textbf{c})=3n$,
    we only need to discuss the cases $4 \leq w_H(\textbf{c})\leq  6$.

    Case I: $w_H(\textbf{c})=4$

    According to Lemma \ref{le3.1}, the support of $\textbf{c}$ of weight $4$ must satisfy 
    $\{(i_1,i_2,i_3,i_4):  jn\leq i_1<i_2<i_3<i_4<(j+1)n,\ j\in [0,2] \} $, 
    then after permutations $\rho$ and $\phi$, we can easily get $I=|\Theta_2(\phi(\rho(\textbf{c})))|= 4$ and 
    $w_{sp}(\phi(\rho(\textbf{c})))=w_H(\textbf{c})+I=8$.

    Case II: $w_H(\textbf{c})=5$
    
    The support of $\textbf{c}$  of weight $5$ must satisfy 
    $\{(i_1,i_2,i_3,i_4,i_5):  jn\leq i_1<i_2<i_3<i_4<i_5<(j+1)n,\ j\in [0,2]   \}$, then $I=5$ and $w_{sp}(\phi(\rho(\textbf{c})))=10$.
   
    Case III: $w_H(\textbf{c})=6$

    According to Lemma \ref{le3.1},
    if the support of $\textbf{c}$ satisfy 
    $\{(i_1,i_2,i_3,i_4,i_5,i_6): jn\leq  i_1<i_2<i_3<i_4<i_5<i_6 <(j+1)n, \ j\in [0,2] \} $ 
    or $\{(i_1,i_2,i_3,i_4,i_5,i_6): j_1n\leq  i_1<i_2<i_3<(j_1+1)n, \ j_2n \leq i_4<i_5<i_6< (j_2+1)n,\ j_1\neq j_2\ and \ j_1,j_2\in [0,2]  \}$,  
    after permutations $\rho$ and $\phi$, we have $I\geq 3$ and $w_{sp}(\phi(\rho(\textbf{c})))\geq 9$.
    
    If the support of $\textbf{c}$ satisfy $\{(i_1,i_2,i_3,i_4,i_5,i_6): j_1n\leq  i_1<i_2<(j_1+1)n, \ j_2n \leq i_3<i_4< (j_2+1)n,\ j_3n \leq i_5<i_6< (j_3+1)n,\  j_1\neq j_2\neq j_3\ and\  j_1,j_2,j_3\in [0,2]\}$,    
    with $\alpha_{i_1}+\alpha_{i_2}=\alpha_{i_3}+\alpha_{i_4}=\alpha_{i_5}+\alpha_{i_6}$, and the subscripts are reduced modulo $n$, 
    then we can get $w_{sp}(\phi(\rho(\textbf{c})))\geq 8$ except there exists a codeword satisfying $(w_H(\textbf{c}),w_{sp}(\phi(\rho(\textbf{c}))))=(6,7)$, 
    which means that there must exist six consecutive nonzero entries after permutations $\rho$ and $\phi$ with the following three cases:
   
    Case III-1: $c_{0,m},c_{1,m+1},c_{2,m},c_{0,m+1},c_{1,m+2},c_{2,m+1}$, where $m\in [0,n-2]$ and the subscripts of $c_{i,j}$ are reduced modulo $n$.

    Since $\textbf{a}$ is a vector which consists of the first $n$ elements of  $\mathcal{F}$ with different elements,
    then we can get $\alpha_{m}\neq \alpha_{m+2}$, which implies that $\alpha_{m}+\alpha_{m+1}\neq \alpha_{m+1}+\alpha_{m+2}$.
    This is a contradiction due to $\alpha_{i_1}+\alpha_{i_2}=\alpha_{i_3}+\alpha_{i_4}=\alpha_{i_5}+\alpha_{i_6}$.

    Case III-2: $c_{1,m+1},c_{2,m},c_{0,m+1},c_{1,m+2},c_{2,m+1},c_{0,m+2}$, where $m\in [0,n-3]$. 

    We can get a similar contradiction because of $\alpha_{m}\neq  \alpha_{m+2}$.

    Case III-3: $c_{2,m},c_{0,m+1},c_{1,m+2},c_{2,m+1},c_{0,m+2},c_{1,m+3}$, where $m\in [0,n-3]$. 

    We can get a similar contradiction because of $\alpha_{m+1} \neq \alpha_{m+2}\neq \alpha_{m+3}$.

    By classification, we find that the existence of $w_{sp}(\phi(\rho(\textbf{c})))=7$ contradicts the fact  $\alpha_{i_1}+\alpha_{i_2}=\alpha_{i_3}+\alpha_{i_4}=\alpha_{i_5}+\alpha_{i_6}$.

    Therefore, by discussing the codewords satisfying $4 \leq w_H(\textbf{c})\leq  6$, respectively,
    we can get $w_{sp}(\mathcal{D})=w_{sp}(\mathcal{\phi(\rho(\textbf{c}))})\geq 8$. 
    According to Lemma \ref{le1.1}, $d_{sp}(\mathcal{D})\leq 3n-(3n-6)+2=8$. 
    So $d_{sp}(\mathcal{D})=8$ and $\mathcal{D}$ is an MDS $(3n,8)_q$ symbol-pair code.
\end{proof}

\begin{remark}
    In \cite{LELP23}, the authors constructed the following two classes of MDS symbol-pair codes from the MP codes:

(1) Suppose that $q$ is a power of an odd prime $p$ with $q\equiv 1\ ({\rm mod}\ 3)$, and  $n=mp$, then for each $m\in [1,\frac{q}{p}]$, 
there exists an MDS symbol-pair code with parameters $(3n,7)_q$. 

(2) Suppose that $q$ is a power of a prime $p$ with $q\equiv 1\ ({\rm mod}\ 3)$, and $n=mp$, then for each $m\in [1,\frac{q}{p}]$, 
there exists an MDS symbol-pair code with parameters $(3n,10)_q$. 

Notably, it is discontinuous for $n$ taking values in the interval $[1,q]$(Surely, the length $\mathcal{N}=3n$ must be guaranteed to be greater than or equal to the symbol-pair distance), i.e., $n$ is intermittent in the interval $[1,q]$.
Does the same MDS symbol-pair codes still exist when $n$ takes consecutive values in the interval  $[1,q]$?
We will discuss it in the following two subsections which improves their conclusions.

\end{remark}

\subsubsection{Symbol-pair distance $d_{sp}(\mathcal{C})=7$ }

Assume that $\mathcal{C}_1=\mathcal{C}_2$ is the GRS code $GRS_{1}$ with parameters $[n,n-1,2]_q$, and
$\mathcal{C}_3$ is the GRS code $GRS_{3}$ with parameters $[n,n-3,4]_q$.

Define the MP code
\begin{align}\label{eq3.3}
    \mathcal{C}=[\mathcal{C}_1,\mathcal{C}_2,\mathcal{C}_3]\cdot A.
\end{align}
Obviously, $\mathcal{C}$ is a $[3n,3n-5]_q$ code. 
The support for the codewords of the MP code $\mathcal{C}$ defined in (\ref{eq3.3}) has the following lemma.

\begin{lemma}\cite{LELP23}\label{le3.2}
    Let $\mathcal{C}$ be the MP code as defined in (\ref{eq3.3}), and $\textbf{c}$ be a codeword of $\mathcal{C}$ with coordinates indexed by the set $[0,3n-1]$.
    Then the following results hold.
    
    (i) There are no codewords in code $\mathcal{C}$ with Hamming weight less than $4$.

    (ii) If $w_H(\textbf{c})=4$, then the support  of $\textbf{c}$ must satisfy 
    $\{(i_1,i_2,i_3,i_4):  jn\leq i_1<i_2<i_3<i_4<(j+1)n,\ j\in [0,2] \} $ 
    or $\{(i_1,i_2,i_3,i_4): j_1n\leq  i_1<i_2<(j_1+1)n, \ j_2n \leq i_3<i_4< (j_2+1)n, \ j_1\neq j_2\ and \ j_1,j_2\in [0,2] \}$ 
    with $\alpha_{i_1}+\alpha_{i_2}=\alpha_{i_3}+\alpha_{i_4}$, where the subscripts are reduced modulo $n$.

    (iii) If $w_H(\textbf{c})=5$, then the support of $\textbf{c}$ must satisfy 
    $\{(i_1,i_2,i_3,i_4,i_5):  jn\leq i_1<i_2<i_3<i_4<i_5<(j+1)n,\ j\in [0,2] \} $ or 
    $\{(i_1,i_2,i_3,i_4,i_5): j_1n\leq  i_1<i_2<i_3<(j_1+1)n, \ j_2n \leq i_4<i_5<(j_2+1)n,\ j_1\neq j_2\ and \ j_1,j_2\in [0,2]  \}$.

\end{lemma}

With the help of the above lemma, the following theorem holds

\begin{theorem}\label{the3.2}
    Suppose that $q$ is a power of an odd prime number $p$ with $q\equiv 1\ ({\rm mod}\ 3)$, then for each $n\in [4,q]$, there exists an MDS $(3n,7)_q$ symbol-pair code.
\end{theorem}

\begin{proof}
    As in terms of cosets, we can write $\mathbb{F}_q \triangleq \bigcup_{i=0}^{\frac{q}{p}-1}(\chi_i+\mathbb{F}_p)$, where $\chi_0=0$. Let 
    $$\mathcal{F}=(0,1,2,\cdots,p-1,\chi_1,\chi_1+1,\cdots,\chi_1+p-1,\cdots \cdots,\chi_{\frac{q}{p}-1},\cdots,\chi_{\frac{q}{p}-1}+p-1 ),$$
    which consists of all the distinct elements in $\mathbb{F}_q$, and let $\textbf{a}=(\alpha_0,\alpha_1,\cdots,\alpha_{n-1})$ be a vector formed by the first $n$ elements of $\mathcal{F}$.
    Suppose that $\mathcal{C}$ is an MP code as defined (\ref{eq3.3}) and $\mathcal{D}$ is a code that permutates equivalent to the code $\mathcal{C}$ under 
    the specific permutations  $\rho$ and $\phi$ (see (\ref{eq3.1})).
    Namely,
    \begin{equation*}
        \begin{aligned}
            \mathcal{D}: =&\phi(\rho(\mathcal{C}))  \\
                       : =&\{\phi(\rho(\textbf{c})), \  \forall \textbf{c}\in \mathcal{C}\}. \\
        \end{aligned}
    \end{equation*}

    The codes $\mathcal{C}$ and $\mathcal{D}$ have the same parameters $[3n,3n-5,4]_q$ because they are permutation equivalent.
    We will illustrate that $d_{sp}(\mathcal{D})=7$ with the help of the support distribution of the codewords $\textbf{c}$ of $\mathcal{C}$.
   
    From Lemma \ref{le3.2}, there are no codewords with Hamming weight less than $4$. 
    Combining with $w_{sp}(\textbf{c})\geq w_H(\textbf{c})+1$ for $w_H(\textbf{c})<3n$, and $w_{sp}(\textbf{c})=w_H(\textbf{c})$ for $w_H(\textbf{c})=3n$,
    we only need to discuss the cases $4 \leq w_H(\textbf{c})\leq  5$.

    Case I: $w_H(\textbf{c})=4$

    If the support of $\textbf{c}$ satisfy 
    $\{(i_1,i_2,i_3,i_4):  jn\leq i_1<i_2<i_3<i_4<(j+1)n,\ j\in [0,2] \} $, then $w_{sp}(\phi(\rho(\textbf{c})))=8$.
    
    If the support of $\textbf{c}$ satisfy 
    $\{(i_1,i_2,i_3,i_4): j_1n\leq  i_1<i_2<(j_1+1)n, \ j_2n \leq i_3<i_4< (j_2+1)n, \ j_1\neq j_2\ and \ j_1,j_2\in [0,2]  \}$ 
    with $\alpha_{i_1}+\alpha_{i_2}=\alpha_{i_3}+\alpha_{i_4}$, where the subscripts are reduced modulo $n$,  
    then after permutations, $w_{sp}(\phi(\rho(\textbf{c})))\geq 7$ except the codewords satisfy the following cases:
   
    Case I-1: $c_{0,m_1},c_{1,m_1+1},c_{0,m_2},c_{1,m_2+1}$, where $m_1\neq m_2$ and $m_1,m_2\in [0,n-1]$. 
    
    We will find contradictions with $\alpha_{i_1}+\alpha_{i_2}=\alpha_{i_3}+\alpha_{i_4}$ (Here it is specified as $\alpha_{m_1+1}+\alpha_{m_2+1}=\alpha_{m_1}+\alpha_{m_2}$).
    It will be stated separately in terms of whether $\alpha_{m_1},\alpha_{m_2}$ belong to a subset $\mathbf{U}$ of $\mathcal{F}$ with the form
    $$\mathbf{U}=\{p-1, \chi_1+p-1, \chi_2+p-1, \cdots, \chi_{\frac{q}{p}-1}+p-1\}.$$

    If $\alpha_{m_1},\alpha_{m_2}\notin \mathbf{U}$,
    then we can get $\alpha_{m_1+1}=\alpha_{m_1}+1$ and $\alpha_{m_2+1}=\alpha_{m_2}+1$, 
    which implies that $\alpha_{m_1+1}+\alpha_{m_2+1}=\alpha_{m_1}+\alpha_{m_2}+2$.
    Since $p$ is odd,  we have $\alpha_{m_1+1}+\alpha_{m_2+1}\neq \alpha_{m_1}+\alpha_{m_2}$, a contradiction.

    If one of $\alpha_{m_1},\alpha_{m_2}$ belongs to $\mathbf{U}$, 
    then it may be assumed that  $\alpha_{m_1}=\chi_i+p-1$, where $i\in \{0,1,2,\cdots ,\frac{q}{p}-1 \} $. 
    Consequently, $\alpha_{m_1+1}=\chi_{i+1}$. Combining with $\alpha_{m_2+1}=\alpha_{m_2}+1$,
    we can get $\alpha_{m_1+1}+\alpha_{m_2+1}=\chi_{i+1}+\alpha_{m_2}+1$.
    Assume that $\alpha_{m_1+1}+\alpha_{m_2+1}=\alpha_{m_1}+\alpha_{m_2}$, 
    then we have $\alpha_{m_1}=\chi_{i+1}+1$, which contradicts $\alpha_{m_1}=\chi_i+p-1$.

    If  $\alpha_{m_1},\alpha_{m_2}$ both belong to  $\mathbf{U}$, 
    then it may be assumed that  $\alpha_{m_1}=\chi_{i_1}+p-1$ and  $\alpha_{m_2}=\chi_{i_2}+p-1$ , where $i_1,i_2\in \{0,1,2,\cdots, \frac{q}{p}-1 \}$. 
    Consequently, $\alpha_{m_1+1}=\chi_{i_1+1}$, $\alpha_{m_2+1}=\chi_{i_2+1}$
    and $\alpha_{m_1+1}+\alpha_{m_2+1}=\chi_{i_1+1}+\chi_{i_2+1}$.
    Assume that $\alpha_{m_1+1}+\alpha_{m_2+1}=\alpha_{m_1}+\alpha_{m_2}$, 
    then we have $\chi_{i_1+1}+\chi_{i_2+1}=\chi_{i_1}+\chi_{i_2}+2p-2=\chi_{i_1}+\chi_{i_2}-2$, 
    which is impossible due to the fact that $p$ is odd.

    Case I-2: $c_{1,m_1+1},c_{2,m_1},c_{1,m_2+1},c_{2,m_2}$, where $m_1\neq m_2$ and $m_1,m_2\in [0,n-1]$. 
    
    Case I-3: $c_{2,m_1},c_{0,m_1+1},c_{2,m_2},c_{0,m_2+1}$, where $m_1\neq m_2$ and $m_1,m_2\in [0,n-2]$. 
    
    Analogous to Case I-1, we can  get the similar contradictions for Cases I-2 and I-3.

    Case II: $w_H(\textbf{c})=5$
    
    The support of $\textbf{c}$ of weight $5$ must  satisfy 
    $\{(i_1,i_2,i_3,i_4,i_5):  jn\leq i_1<i_2<i_3<i_4<i_5<(j+1)n,\ j\in [0,2]  \} $ or
    $\{(i_1,i_2,i_3,i_4,i_5): j_1n\leq  i_1<i_2<i_3<(j_1+1)n, \ j_2n \leq i_4<i_5<(j_2+1)n,\ j_1\neq j_2\ and \ j_1,j_2\in [0,2] \} $.   
    After permutations, we can all get $I\geq 3$ and $w_{sp}(\phi(\rho(\textbf{c})))\geq 8$.

    Therefore, by discussing the codewords satisfying $4 \leq w_H(\textbf{c})\leq  5$, respectively, we can get $w_{sp}(\mathcal{D})=w_{sp}(\mathcal{\phi(\rho(\textbf{c}))})\geq 7$. 
    According to Lemma \ref{le1.1}, $d_{sp}(\mathcal{D})\leq 3n-(3n-5)+2=7$. 
    So $d_{sp}(\mathcal{D})=7$ and $\mathcal{D}$ is an MDS $(3n,7)_q$ symbol-pair code.

\end{proof}

\subsubsection{Symbol-pair distance $d_{sp}(\mathcal{C})=10$ }

Assume that 
$\mathcal{C}_1=\mathcal{C}_2$ is the GRS code $GRS_{2}$ with parameters $[n,n-2,3]_q$, 
$\mathcal{C}_3$ is the GRS code $GRS_{4}$ with parameters $[n,n-4,5]_q$.     

Define the MP code
\begin{align}\label{eq3.4}
    \mathcal{C}=[\mathcal{C}_1,\mathcal{C}_2,\mathcal{C}_3]\cdot A.
\end{align}
Obviously, $\mathcal{C}$ is a $[3n,3n-8]_q$ code. The support for the codewords of the MP code $\mathcal{C}$ defined in (\ref{eq3.4}) has the following lemma.

\begin{lemma}\cite{LELP23}\label{le3.3}
    Let $\mathcal{C}$ be the MP code as defined in (\ref{eq3.4}), and $\textbf{c}$ be a codeword of $\mathcal{C}$ with coordinates indexed by the set $[0,3n-1]$.
    Then the following results hold.
    
    (i) There are no codewords in code $\mathcal{C}$ with Hamming weight less than $5$.

    (ii) If $w_H(\textbf{c})=5$, then the support of $\textbf{c}$ must satisfy 
    $\{(i_1,i_2,i_3,i_4,i_5):  jn\leq i_1<i_2<i_3<i_4<i_5<(j+1)n,\ j\in[0,2]  \} $.
    
    (iii) If $w_H(\textbf{c})=6$, then the support of $\textbf{c}$ must satisfy 
    $\{(i_1,i_2,i_3,i_4,i_5,i_6): jn\leq  i_1<i_2<i_3<i_4<i_5<i_6 <(j+1)n, \ j\in[0,2] \} $ 
    or $\{(i_1,i_2,i_3,i_4,i_5,i_6): j_1n\leq  i_1<i_2<i_3<(j_1+1)n, \ j_2n \leq i_4<i_5<i_6< (j_2+1)n, \ j_1\neq j_2\ and \ j_1,j_2\in [0,2] \} $  
    with $\alpha_{i_1}+\alpha_{i_2}+\alpha_{i_3}=\alpha_{i_4}+\alpha_{i_5}+\alpha_{i_6}$, where the subscripts are reduced modulo $n$.  

    (iv) If $w_H(\textbf{c})=7$, then the support of $\textbf{c}$ must satisfy 
    $\{(i_1,i_2,i_3,i_4,i_5,i_6,i_7): jn\leq  i_1<i_2<i_3<i_4<i_5<i_6<i_7 <(j+1)n, \ j\in[0,2]  \} $ 
    or $\{(i_1,i_2,i_3,i_4,i_5,i_6,i_7): j_1n\leq  i_1<i_2<i_3<i_4<(j_1+1)n, \ j_2n \leq i_5<i_6<i_7< (j_2+1)n, \ j_1\neq j_2\ and \ j_1,j_2\in [0,2] \}$.

    (v) If $w_H(\textbf{c})=8$, then the support of $\textbf{c}$ must satisfy 
    $\{(i_1,i_2,i_3,i_4,i_5,i_6,i_7,i_8): jn\leq  i_1<i_2<i_3<i_4<i_5<i_6<i_7<i_8<(j+1)n, \ j\in [0,2] \} $ or
    $\{(i_1,i_2,i_3,i_4,i_5,i_6,i_7,i_8): j_1n\leq  i_1<i_2<i_3<i_4<i_5<(j_1+1)n, \ j_2n \leq i_6<i_7<i_8< (j_2+1)n, \ j_1\neq j_2\ and \ j_1,j_2\in [0,2]\} $ or
    $\{(i_1,i_2,i_3,i_4,i_5,i_6,i_7,i_8): j_1n\leq  i_1<i_2<i_3<i_4<(j_1+1)n, \ j_2n \leq i_5<i_6<i_7<i_8<(j_2+1)n, \ j_1\neq j_2\ and \ j_1,j_2\in [0,2]\} $.

\end{lemma}

With the help of the above lemma, the following theorem holds.
        
\begin{theorem}\label{the3.3}
    Suppose that $q$ is a power of a prime number $p$ with $q\equiv 1\ ({\rm mod}\ 3)$, then for each $n\in [5,q]$, there exists an MDS $(3n,10)_q$ symbol-pair code.
\end{theorem}

\begin{proof}
    As $\mathbb{F}_q$ can be written as $\mathbb{F}_q \triangleq \bigcup_{i=0}^{\frac{q}{p}-1}(\chi_i+\mathbb{F}_p)$ in terms of cosets, where $\chi_0=0$. Let 
    $$\mathcal{F}=(0,1,2,\cdots,p-1,\chi_1,\chi_1+1,\cdots,\chi_1+p-1,\cdots \cdots,\chi_{\frac{q}{p}-1},\cdots,\chi_{\frac{q}{p}-1}+p-1 ),$$
   which consists of all the distinct elements in $\mathbb{F}_q$, and let  $\textbf{a}=(\alpha_0,\alpha_1,\cdots,\alpha_{n-1})$  be a vector formed by the first $n$ elements of $\mathcal{F}$.
    Suppose that $\mathcal{C}$ is an MP code as defined in (\ref{eq3.4}) and $\mathcal{D}$ is a code that permutates equivalent to the code $\mathcal{C}$ under 
    the specific permutations $\rho$ and $\phi$ (see (\ref{eq3.1})).
    Namely,
    \begin{equation*}
        \begin{aligned}
            \mathcal{D}: =&\phi(\rho(\mathcal{C}))  \\
                       : =&\{\phi(\rho(\textbf{c})), \  \forall \textbf{c}\in \mathcal{C}\}. \\
        \end{aligned}
    \end{equation*}

    The codes $\mathcal{C}$ and $\mathcal{D}$ have the same parameters $[3n,3n-8,5]_q$ because they are permutation equivalent.
    We will illustrate that $d_{sp}(\mathcal{D})=10$ with the help of the support distribution of the codewords $\textbf{c}$ of $\mathcal{C}$.

    From Lemma \ref{le3.3}, there are no codewords with Hamming weight less than $5$. 
    Combining with $w_{sp}(\textbf{c})\geq w_H(\textbf{c})+1$ for $w_H(\textbf{c})<3n$, and $w_{sp}(\textbf{c})=w_H(\textbf{c})$ for $w_H(\textbf{c})=3n$,
    we only need to discuss the cases $5 \leq w_H(\textbf{c})\leq  8$.

    Case I: $w_H(\textbf{c})=5$
    
    According to Lemma \ref{le3.3}, the support of $\textbf{c}$ satisfy 
    $\{(i_1,i_2,i_3,i_4,i_5):  jn\leq i_1<i_2<i_3<i_4<i_5<(j+1)n,\ j\in [0,2] \} $,  
    then we can easily obtain $w_{sp}(\phi(\rho(\textbf{c})))=10$.

    Case II: $w_H(\textbf{c})=6$
    
    If the support of $\textbf{c}$ satisfy
    $\{(i_1,i_2,i_3,i_4,i_5,i_6): jn\leq  i_1<i_2<i_3<i_4<i_5<i_6 <(j+1)n, \ j\in [0,2] \} $,
    then we can easily obtain $w_{sp}(\phi(\rho(\textbf{c})))=12$.

    If the support of $\textbf{c}$ satisfy 
    $\{(i_1,i_2,i_3,i_4,i_5,i_6): j_1n\leq  i_1<i_2<i_3<(j_1+1)n, \ j_2n \leq i_4<i_5<i_6< (j_2+1)n,\ j_1\neq j_2\ and \ j_1,j_2\in [0,2] \} $  
    with $\alpha_{i_1}+\alpha_{i_2}+\alpha_{i_3}=\alpha_{i_4}+\alpha_{i_5}+\alpha_{i_6}$, where the subscripts are reduced modulo $n$,
    then $w_{sp}(\phi(\rho(\textbf{c})))\geq 10$ except the codewords satisfy the following cases:

    Case II-1: $c_{0,m_1},c_{1,m_1+1},c_{0,m_2},c_{1,m_2+1},c_{0,m_3},c_{1,m_3+1}$, where $m_1\neq m_2 \neq m_3$ and $m_1,m_2,m_3\in [0,n-1]$. 
    
    Let a subset $\mathbf{U}$ of $\mathcal{F}$ be
     $$\mathbf{U}=\{p-1, \chi_1+p-1, \chi_2+p-1, \cdots, \chi_{\frac{q}{p}-1}+p-1\}.$$

    If $\alpha_{m_1},\alpha_{m_2},\alpha_{m_3}\notin \mathbf{U}$,
    then we can get $\alpha_{m_1+1}=\alpha_{m_1}+1$, $\alpha_{m_2+1}=\alpha_{m_2}+1$ and $\alpha_{m_3+1}=\alpha_{m_3}+1$. 
    Consequently, $\alpha_{m_1+1}+\alpha_{m_2+1}+\alpha_{m_3+1}=\alpha_{m_1}+\alpha_{m_2}+\alpha_{m_3}+3$.
    Since $q$ is a power of a prime number $p$ with $q\equiv 1\ ({\rm mod}\ 3)$, then there must be $p\neq 3$ and
     $\alpha_{m_1+1}+\alpha_{m_2+1}+\alpha_{m_3+1}\neq \alpha_{m_1}+\alpha_{m_2}+\alpha_{m_3}$.
    This is a contradiction due to  $\alpha_{i_1}+\alpha_{i_2}+\alpha_{i_3}=\alpha_{i_4}+\alpha_{i_5}+\alpha_{i_6}$.

    If one of $\alpha_{m_1},\alpha_{m_2},\alpha_{m_3}$  belongs to $\mathbf{U}$, 
    then it may be assumed that  $\alpha_{m_1}=\chi_i+p-1$, where $i\in \{0,1,2,\cdots, \frac{q}{p}-1 \}$. 
    Consequently, $\alpha_{m_1+1}=\chi_{i+1}$. Combining with $\alpha_{m_2+1}=\alpha_{m_2}+1$, $\alpha_{m_3+1}=\alpha_{m_3}+1$,
    we can get $\alpha_{m_1+1}+\alpha_{m_2+1}+\alpha_{m_3+1}=\chi_{i+1}+\alpha_{m_2}+\alpha_{m_3}+2$.
    Assume that  $\alpha_{m_1+1}+\alpha_{m_2+1}+\alpha_{m_3+1}=\alpha_{m_1}+\alpha_{m_2}+\alpha_{m_3}$,
    then we have $\alpha_{m_1}=\chi_{i+1}+2$, which  contradicts $\alpha_{m_1}=\chi_i+p-1$.

    If two of $\alpha_{m_1},\alpha_{m_2},\alpha_{m_3}$  belong to  $\mathbf{U}$, 
    then it may be assumed that  $\alpha_{m_1}=\chi_{i_1}+p-1$, $\alpha_{m_2}=\chi_{i_2}+p-1$, where $i_1,i_2\in \{0,1,2,\cdots ,\frac{q}{p}-1 \}$. 
    Consequently, $\alpha_{m_1+1}=\chi_{i_1+1}$ and $\alpha_{m_2+1}=\chi_{i_2+1}$.
    Combining with $\alpha_{m_3+1}=\alpha_{m_3}+1$,
    we can get $\alpha_{m_1+1}+\alpha_{m_2+1}+\alpha_{m_3+1}=\chi_{i_1+1}+\chi_{i_2+1}+\alpha_{m_3}+1$.
    Assume that  $\alpha_{m_1+1}+\alpha_{m_2+1}+\alpha_{m_3+1}=\alpha_{m_1}+\alpha_{m_2}+\alpha_{m_3}$,
    then we have $\alpha_{m_1}+\alpha_{m_2}=\chi_{i_1+1}+\chi_{i_2+1}+1$.
    Since $\alpha_{m_1}+\alpha_{m_2}=\chi_{i_1}+\chi_{i_2}-2$, then we have  
    $\chi_{i_1+1}+\chi_{i_2+1}=\chi_{i_1}+\chi_{i_2}-3$.
    This is a contradiction because  $p\neq 3$ is a prime number.

    If  $\alpha_{m_1},\alpha_{m_2},\alpha_{m_3}$ all belong to  $\mathbf{U}$, 
    then it may be assumed that  $\alpha_{m_1}=\chi_{i_1}+p-1$, $\alpha_{m_2}=\chi_{i_2}+p-1$ and $\alpha_{m_3}=\chi_{i_3}+p-1$, where $i_1,i_2,i_3\in \{0,1,2,\cdots, \frac{q}{p}-1 \}$. 
    Consequently, $\alpha_{m_1+1}=\chi_{i_1+1}$, $\alpha_{m_2+1}=\chi_{i_2+1}$ and $\alpha_{m_3+1}=\chi_{i_3+1}$.
    Obviously, $\alpha_{m_1+1}+\alpha_{m_2+1}+\alpha_{m_3+1}=\chi_{i_1+1}+\chi_{i_2+1}+\chi_{i_3+1}$.
    Assume that  $\alpha_{m_1+1}+\alpha_{m_2+1}+\alpha_{m_3+1}=\alpha_{m_1}+\alpha_{m_2}+\alpha_{m_3}$,
    then we have $\chi_{i_1+1}+\chi_{i_2+1}+\chi_{i_3+1}=\chi_{i_1}+\chi_{i_2}+\chi_{i_3}+3p-3=\chi_{i_1}+\chi_{i_2}+\chi_{i_3}-3$.
    This is also a contradiction because $p\neq 3$ is a prime number.

    Case II-2: $c_{1,m_1+1},c_{2,m_1},c_{1,m_2+1},c_{2,m_2},c_{1,m_3+1},c_{2,m_3}$, where $m_1\neq m_2 \neq m_3$ and $m_1,m_2,m_3\in [0,n-1]$. 
    
    Case II-3: $c_{2,m_1},c_{0,m_1+1},c_{2,m_2},c_{0,m_2+1},c_{2,m_3},c_{0,m_3+1}$, where $m_1\neq m_2 \neq m_3$ and $m_1,m_2,m_3\in [0,n-2]$. 
    
    Analogous to Case II-1, we can obtain the similar contradictions for Cases II-2 and II-3.

    Case III: $w_H(\textbf{c})=7$ or $8$

    Based on Lemma \ref{le3.3}, for any the support distribution of $\textbf{c}$ with weights $7$ and $8$, 
    it is easy to see that $w_{sp}(\phi(\rho(\textbf{c})))\geq 10$.
    
    Therefore, by discussing the codewords satisfying $5 \leq w_H(\textbf{c})\leq  8$, respectively, we can get $w_{sp}(\mathcal{D})=w_{sp}(\mathcal{\phi(\rho(\textbf{c}))})\geq 10$.
    According to Lemma \ref{le1.1}, $d_{sp}(\mathcal{D})\leq 3n-(3n-8)+2=10$. 
    So $d_{sp}(\mathcal{D})=10$ and $\mathcal{D}$ is an MDS $(3n,10)_q$ symbol-pair code.

\end{proof}

In the following, for Theorems 3.1-3.3, we will respectively give an example.

\begin{example}\label{example3.1}
    Let $p=2$, $q=4$, and $\mathbb{F}_4=\{0,1,\alpha,\alpha+1\}$, where $\alpha$ is a root of the irreducible polynomial $f(x)=x^2+x+1$ over $\mathbb{F}_2$.
    Let $\xi=\alpha$ be  a primitive element of $\mathbb{F}_4$, then $\xi^2=\alpha+1$, $\xi^3=1$.
    Taking $\omega=\xi$ be a primitive $3$-th root of unity and 
    $$ A=\begin{pmatrix}
    1   & 1   & 1    \\
    1   & \xi  & \xi^2    \\
    1   & \xi^2   & \xi  \\
    \end{pmatrix}.$$
Then we can get 
    $$ (A^{-1})^T=\frac{1}{3}\begin{pmatrix}
        1   & 1   & 1\\
        1   &\xi^2   & \xi  \\
        1   & \xi   &\xi^2  \\
    \end{pmatrix}.$$ 
    
    Assume that $\mathcal{C}_1$ is the GRS code with parameters $[4,3,2]_4$ whose parity-check matrix is 
    $$ H_1=\begin{pmatrix}
    1   & 1   & 1 &  1 \\
    \end{pmatrix},$$
    $\mathcal{C}_2$ is the GRS code  with parameters $[4,2,3]_4$ whose parity-check matrix is 
    $$ H_2=\begin{pmatrix}
        1   & 1   & 1 &  1  \\
        0   & \xi   & \xi^2 &  1\\
    \end{pmatrix},$$
    $\mathcal{C}_3$ is the GRS code with parameters $[4,1,4]_4$ whose parity-check matrix is 
    $$ H_3=\begin{pmatrix}
        1   & 1   & 1 &  1 \\
        0   & \xi    & \xi^2 &  1\\
        0   & \xi^2   & \xi &  1 \\
    \end{pmatrix}.$$

    The parity-check matrix of the MP code $\mathcal{C}=[\mathcal{C}_1,\mathcal{C}_2,\mathcal{C}_3]\cdot A$ with parameters $[12,6,4]_4$ is 
    \setcounter{MaxMatrixCols}{50}
    $$ H=\begin{pmatrix}
        1   & 1      & 1     &  1 &  1    & 1     & 1     &  1     &  1   & 1     & 1     &  1  \\
        1   & 1      & 1     &  1 & \xi^2 & \xi^2 & \xi^2 &  \xi^2 &  \xi & \xi   & \xi   & \xi  \\
        0   & \xi    & \xi^2 &  1 & 0     & 1     & \xi   &  \xi^2 & 0    & \xi^2 & 1     &  \xi   \\
        1   & 1      & 1     &  1 & \xi   & \xi   & \xi   &  \xi   & \xi^2& \xi^2 & \xi^2 &  \xi^2\\
        0   & \xi    & \xi^2 &  1 & 0     & \xi^2 & 1     & \xi    & 0    & 1     & \xi   &  \xi^2\\
        0   & \xi^2  & \xi   &  1 & 0     & 1     & \xi^2 &  \xi   & 0    & \xi   & 1     & \xi^2  \\
    \end{pmatrix}.$$

    After the following specific permutations, 
   
    $$1,2,3,4,5,6,7,8,9,10,11,12$$
    $$\downarrow\ \rho$$
    $$1,5,9,2,6,10,3,7,11,4,8,12$$
    $$\downarrow\ \phi$$
    $$1,6,9,2,7,10,3,8,11,4,5,12,$$
   
    the code $\mathcal{D}$ with parameters $[12,6,4]_4$ has a parity-check matrix of the form:
    $$ H=\begin{pmatrix}
        1   & 1      & 1     &  1   &  1    & 1     & 1     &  1     &  1   & 1     & 1     &  1  \\
        1   & \xi^2  & \xi   &  1   & \xi^2 & \xi   & 1     &  \xi^2 &  \xi & 1   & \xi^2   & \xi  \\
        0   & 1      & 0     &  \xi & \xi   & \xi^2 & \xi^2 &  \xi^2 & 1    & 1 & 0     &  \xi   \\
        1   & \xi    &\xi^2  &  1   & \xi   & \xi^2 & 1     &  \xi   & \xi^2& 1 & \xi &  \xi^2\\
        0   &  \xi^2 & 0     &  \xi & 1     & 1     & \xi^2 & \xi    & \xi   & 1     & 0   &  \xi^2\\
        0   &1       & 0     & \xi^2  & \xi^2 & \xi   & \xi   &  \xi   & 1    & 1   & 0     & \xi^2  \\
    \end{pmatrix}.$$

    By using Magma, $\mathcal{D}$ has a generator matrix 
    $$ G=\begin{pmatrix}
        1  &  0  &  0  &  0  &  0  &  0  &  1 & \xi^2  & \xi    &0 & \xi^2  & \xi \\
        0  &  1  &  0  &  0  &  0 &   0 &   1 &   0  &  1  &  1 &   1  &  1 \\
        0  &  0   & 1  &  0  &  0  &  0  &  1 &  \xi &  \xi  &  1  & \xi &\xi^2 \\
        0  &  0  &  0  &  1  &  0  &  0  &  0& \xi^2  & \xi  &  1 & \xi^2  & \xi \\
        0  &   0  &  0 &   0  &  1  &  0 &   1 &   1  &  1   & 1  &  0   & 1 \\
        0  &  0 &   0  &  0  &  0  &  1 &   1 &  \xi &\xi^2  &  1  &  \xi  & \xi \\
    \end{pmatrix}.$$
    According to Theorem \ref{the3.1} and the Magma program, the code $\mathcal{D}$ is an MDS $(12,8)_4$ symbol-pair code.

\end{example}

\begin{example}\label{example3.2}
    Let $q=p=7$, and $\mathbb{F}_7=\{0,1,2,3,4,5,6\}$.
    Taking $\omega=2$ be a primitive $3$-th root of unity and 
    $$ A=\begin{pmatrix}
        1   & 1   & 1 \\
        1   & 2   & 4    \\
        1   & 4   & 2  \\
        \end{pmatrix}.$$

Then we can get 
    $$ (A^{-1})^T=\frac{1}{3}\begin{pmatrix}
        1   & 1   & 1  \\
        1   & 4   & 2  \\
        1   & 2   &4    \\
    \end{pmatrix}.$$ 

    Assume that $\mathcal{C}_1=\mathcal{C}_2$ is the GRS code with parameters $[4,3,2]_7$
    whose parity-check matrix is 
    $$ H_1=\begin{pmatrix}
        1   & 1   & 1 &  1\\
    \end{pmatrix},$$
  
    $\mathcal{C}_3$ is the GRS code with parameters $[4,1,4]_7$  whose parity-check matrix is 
    $$ H_2=\begin{pmatrix}
        1   & 1   & 1 &  1\\
        0   & 1   & 2 &  3  \\
        0   & 1   & 4 &  2  \\
    \end{pmatrix}.$$

    The parity-check matrix of the MP code $\mathcal{C}=[\mathcal{C}_1,\mathcal{C}_2,\mathcal{C}_3]\cdot A$ with parameters $[12,7,4]_7$ is 
    \setcounter{MaxMatrixCols}{50}
    $$ H=\frac{1}{3}\begin{pmatrix}
        1   & 1   & 1 &  1     &    1 & 1 & 1 & 1     &    1 & 1 & 1 & 1     \\
        1   & 1   & 1 &  1     &    4 & 4 & 4 & 4     &    2 & 2 & 2 & 2     \\
        1   & 1   & 1 &  1     &    2 & 2 & 2 & 2     &    4 & 4 & 4 & 4     \\
        0   & 1   & 2 &  3     &    0 & 2 & 4 & 6     &    0 & 4 & 1 & 5    \\
        0   & 1   & 4 &  2     &    0 & 2 & 1 & 4     &    0 & 4 & 2 & 1     \\
    \end{pmatrix}.$$

    After the following specific permutations, 
   
    $$1,2,3,4,5,6,7,8,9,10,11,12$$
    $$\downarrow\ \rho$$
    $$1,5,9,2,6,10,3,7,11,4,8,12$$
    $$\downarrow\ \phi$$
    $$1,6,9,2,7,10,3,8,11,4,5,12,$$
   
    the code $\mathcal{D}$ with parameters $[12,7,4]_7$ has a parity-check matrix of the form:
    $$ H=\frac{1}{3}\begin{pmatrix}
        1   & 1   & 1 &  1     &    1 & 1 & 1 & 1     &    1 & 1 & 1 & 1     \\
        1   & 4   & 2 &  1     &    4 & 2 & 1 & 4     &    2 & 1 & 4 & 2     \\
        1   & 2   & 4 &  1     &    2 & 4 & 1 & 2     &    4 & 1 & 2 & 4     \\
        0   & 2   & 0 &  1     &    4 & 4 & 2 & 6     &    1 & 3 & 0 & 5    \\
        0   & 2   & 0 &  1     &    1 & 4 & 4 & 4     &    2 & 2 & 0 & 1     \\
    \end{pmatrix}.$$

    By using Magma, $\mathcal{D}$ has a generator matrix

    $$ G=\frac{1}{3}\begin{pmatrix}
        1 & 0 & 0 & 0 & 0 & 0 & 0 & 4 & 0 & 6 & 3 & 0\\
        0 & 1 & 0 & 0 & 0 & 0 & 0 & 4 & 3 & 0 & 2 & 4\\
        0 & 0 & 1 & 0 & 0 & 0 & 0 & 2 & 0 & 0 & 5 & 6\\
        0 & 0 & 0 & 1 & 0 & 0 & 0 & 6 & 5 & 6 & 1 & 2\\
        0 & 0 & 0 & 0 & 1 & 0 & 0 & 6 & 3 & 0 & 0 & 4\\
        0 & 0 & 0 & 0 & 0 & 1 & 0 & 3 & 6 & 0 & 4 & 0\\
        0 & 0 & 0 & 0 & 0 & 0 & 1 & 0 & 5 & 6 & 0 & 2\\
    \end{pmatrix}.$$

    According to Theorem \ref{the3.2} and the Magma program, the code $\mathcal{D}$ is an MDS $(12,7)_7$ symbol-pair code.

\end{example}

\begin{example}\label{example3.3}
    Let $q=p=7$, and $A$ be defined as in Example \ref{example3.2}.
    Assume that 
$\mathcal{C}_1=\mathcal{C}_2$ is the GRS code with parameters $[5,3,3]_7$ whose parity-check matrix is 
$$ H_1=\begin{pmatrix}
    1   & 1   & 1 &  1 &  1  \\
    0   & 1   & 2 &  3&  4 \\
\end{pmatrix},$$

$\mathcal{C}_3$ is the GRS code with parameters $[5,1,5]_7$ whose parity-check matrix is 
$$ H_2=\begin{pmatrix}
    1   & 1   & 1 &  1  &  1  \\
    0   & 1   & 2 &  3  &  4 \\
    0   & 1   & 4 &  2  &  2 \\
    0   & 1   & 1 &  6  &  1 \\
\end{pmatrix}.$$

The parity-check matrix of the MP code $\mathcal{C}=[\mathcal{C}_1,\mathcal{C}_2,\mathcal{C}_3]\cdot A$ with parameters $[15,7,5]_7$ is 
\setcounter{MaxMatrixCols}{50}
$$ H=\frac{1}{3}\begin{pmatrix}
    1   & 1   & 1 &  1 & 1     &    1 & 1 & 1 & 1 & 1     &    1 & 1 & 1 & 1 & 1    \\
    0   & 1   & 2 &  3 & 4     &    0 & 1 & 2 & 3 & 4     &    0 & 1 & 2 & 3 & 4    \\
    1   & 1   & 1 &  1 & 1     &    4 & 4 & 4 & 4 & 4     &    2 & 2 & 2 & 2 & 2    \\
    0   & 1   & 2 &  3 & 4     &    0 & 4 & 1 & 5 & 2     &    0 & 2 & 4 & 6 & 1   \\

    1   & 1   & 1 &  1 & 1     &    2 & 2 & 2 & 2 & 2     &    4 & 4 & 4 & 4 & 4   \\
    0   & 1   & 2 &  3 & 4     &    0 & 2 & 4 & 6 & 1     &    0 & 4 & 1 & 5 & 2   \\
    0   & 1   & 4 &  2 & 2     &    0 & 2 & 1 & 4 & 4     &    0 & 4 & 2 & 1 & 1   \\
    0   & 1   & 1 &  6 & 1     &    0 & 2 & 2 & 5 & 2     &    0 & 4 & 4 & 3 & 4   \\

\end{pmatrix}.$$

After the following specific permutations, 
   
$$1,2,3,4,5,6,7,8,9,10,11,12,13,14,15$$
$$\downarrow\ \rho$$
$$1,6,11,2,7,12,3,8,13,4,9,14,5,10,15$$
$$\downarrow\ \phi$$
$$1,7,11,2,8,12,3,9,13,4,10,14,5,6,15,$$

the code $\mathcal{D}$ with parameters $[15,7,5]_7$ has a parity-check matrix of the form:

$$ H=\frac{1}{3}\begin{pmatrix}
    1   & 1   & 1 &  1 & 1     &    1 & 1 & 1 & 1 & 1     &    1 & 1 & 1 & 1 & 1    \\
    0   & 1   & 0 &  1 & 2     &    1 & 2 & 3 & 2 & 3     &    4 & 3 & 4 & 0 & 4    \\
    1   & 4   & 2 &  1 & 4     &    2 & 1 & 4 & 2 & 1     &    4 & 2 & 1 & 4 & 2    \\
    0   & 4   & 0 &  1 & 1     &    2 & 2 & 5 & 4 & 3     &    2 & 6 & 4 & 0 & 1   \\

    1   & 2   & 4 &  1 & 2     &    4 & 1 & 2 & 4 & 1     &    2 & 4 & 1 & 2 & 4   \\
    0   & 2   & 0 &  1 & 4     &    4 & 2 & 6 & 1 & 3     &    1 & 5 & 4 & 0 & 2   \\
    0   & 2   & 0 &  1 & 1     &    4 & 4 & 4 & 2 & 2     &    4 & 1 & 2 & 0 & 1   \\
    0   & 2   & 0 &  1 & 2     &    4 & 1 & 5 & 4 & 6     &    2 & 3 & 1 & 0 & 4   \\

\end{pmatrix}.$$
   
By using Magma, $\mathcal{D}$ has a generator matrix

$$ G=\frac{1}{3}\begin{pmatrix}
    1 & 0 & 0 & 0 & 0 & 0 & 0 & 2 & 0 & 3 & 2 & 0 & 3 & 3 & 0\\
    0 & 1 & 0 & 0 & 0 & 0 & 0 & 5 & 1 & 0 & 3 & 5 & 0 & 5 & 1\\
    0 & 0 & 1 & 0 & 0 & 0 & 0 & 1 & 0 & 0 & 1 & 3 & 0 & 5 & 3\\
    0 & 0 & 0 & 1 & 0 & 0 & 0 & 4 & 4 & 4 & 4 & 6 & 2 & 6 & 4\\
    0 & 0 & 0 & 0 & 1 & 0 & 0 & 5 & 3 & 0 & 1 & 1 & 0 & 0 & 3\\
    0 & 0 & 0 & 0 & 0 & 1 & 0 & 2 & 2 & 0 & 2 & 0 & 0 & 3 & 4\\
    0 & 0 & 0 & 0 & 0 & 0 & 1 & 0 & 5 & 5 & 0 & 4 & 1 & 0 & 5\\

\end{pmatrix}.$$

According to Theorem \ref{the3.3} and the Magma program, the code $\mathcal{D}$ is an MDS $(15,10)_7$ symbol-pair code.

\end{example}

\subsection{MP codes with the square matrix $A$ of order $4$}

In this subsection, we will construct symbol-pair codes of length $\mathcal{N}=4n$ from MP codes with the square matrix $A$ of order $4$. 
Let $4\vert(q-1)$, then there must exist a primitive $4$-th root of unity $\omega$ in $\mathbb{F}_q$.
Suppose that $A$ is a $4\times 4$ NSC matrix with the following form:   
$$ A=\begin{pmatrix}
    1   & 1   & 1  & 1  \\
    1   & \omega   & \omega^2 & \omega^3  \\
    1   & \omega^2    & 1 & \omega^2    \\
    1   & \omega^3    & \omega^2  & \omega    \\
\end{pmatrix}.$$

Obviously, we can get 
$$ (A^{-1})^T=\frac{1}{4}\begin{pmatrix}
    1   & 1   & 1 & 1 \\
    1   &\omega^3   & \omega^2   & \omega\\
    1   & \omega^2   &1 & \omega^2  \\
    1   & \omega   &\omega^2 & \omega^3  \\
\end{pmatrix}.$$ 

In the following, we will construct a class of MDS symbol-pair codes and a class of AMDS symbol-pair codes of length $4n$ from the permutation equivalence codes of matrix-product codes.

\subsubsection{Symbol-pair distance $d_{sp}(\mathcal{C})=6$ }

Assume that $\mathcal{C}_1=\mathcal{C}_2$ is the code with parameters $[n,n,1]_q$,
$\mathcal{C}_3$ is the GRS code $GRS_{1}$ with parameters $[n,n-1,2]_q$,
$\mathcal{C}_4$ is the GRS code $GRS_{3}$ with parameters $[n,n-3,4]_q$.

Define the MP code
\begin{align}\label{eq3.5}
    \mathcal{C}=[\mathcal{C}_1,\mathcal{C}_2,\mathcal{C}_3,\mathcal{C}_4]\cdot A.
\end{align}
Obviously, $\mathcal{C}$ is a $[4n,4n-4]_q$ code, and from Corollary \ref{co3.1}, the parity-check matrix for $\mathcal{C}$ is shown below: 

 \[\begin{split} H &=\frac{1}{4}\begin{pmatrix}
    H_1   & \omega^2 H_1  & H_1 & \omega^2 H_1\\
    H_3   & \omega H_3  & \omega^2H_3 & \omega^3 H_3\\
\end{pmatrix},\\
\setcounter{MaxMatrixCols}{14}
&=\frac{1}{4}\begin{pmatrix}
    1             & \cdots       & 1             &\omega^2         & \cdots     & \omega^2           & 1         & \cdots    & 1         & \omega^2   & \cdots  & \omega^2   \\
    1             & \cdots       & 1            &    \omega          & \cdots    & \omega              &    \omega^2        & \cdots    & \omega^2     &    \omega^3        & \cdots         & \omega^3     \\                                             
    \alpha_0      & \cdots       & \alpha_{n-1}    &  \omega\alpha_0   & \cdots    & \omega\alpha_{n-1}  &  \omega^2\alpha_0    & \cdots & \omega^2\alpha_{n-1}  &    \omega^3\alpha_0  & \cdots    & \omega^3\alpha_{n-1}  \\
    \alpha^2_0    & \cdots       & \alpha^2_{n-1}  &   \omega\alpha^2_0& \cdots   &  \omega\alpha^2_{n-1} &   \omega^2\alpha^2_0    & \cdots   &  \omega^2\alpha^2_{n-1} &    \omega^3\alpha^2_0     & \cdots  & \omega^3\alpha^2_{n-1}  \\
\end{pmatrix}.
 \end{split}\]

The following lemma determines the support for codewords in $\mathcal{C}$ whose Hamming weights do not exceed $4$.

\begin{lemma}\label{le3.4}
    Let $\mathcal{C}$ be the MP code as defined in (\ref{eq3.5}), and $\textbf{c}$ be a codeword of $\mathcal{C}$ with coordinates indexed by the set $[0,4n-1]$.
    Then the following results hold.
    
    (i) There are no codewords in code $\mathcal{C}$ with Hamming weight less than $3$.

    (ii) If $w_H(\textbf{c})=3$, then the support  of $\textbf{c}$ must satisfy 
    $\{(i_1,i_2,i_3): j_1n\leq  i_1<(j_1+1)n, \  j_2n\leq  i_2<(j_2+1)n, \  j_3n\leq i_3<(j_3+1)n,\ j_1\neq j_2\neq j_3 \ and\  j_1,j_2,j_3\in [0,3] \}$ 
    with $\alpha_{i_1}=\alpha_{i_2}=\alpha_{i_3}$, where the subscripts are reduced modulo $n$.

    (ii) If $w_H(\textbf{c})=4$, then the support  of $\textbf{c}$ must satisfy 
    $\{(i_1,i_2,i_3,i_4):  jn\leq i_1<i_2<i_3<i_4<(j+1)n,\ j\in [0,3] \} $ or
    $\{(i_1,i_2,i_3,i_4): j_1n\leq  i_1<i_2<(j_1+1)n, \ j_2n \leq i_3<i_4< (j_2+1)n,\  j_1\neq j_2\ and\ j_1,j_2 \in [0,3] \}$    
    or $\{(i_1,i_2,i_3,i_4): j_1n\leq  i_1<i_2<(j_1+1)n, \  j_2n\leq  i_3<(j_2+1)n, \  j_3n\leq i_4<(j_3+1)n, \ j_1\neq j_2\neq j_3\neq j_4 \ and\  j_1,j_2,j_3,j_4\in [0,3] \}$ 
    or $\{(i_1,i_2,i_3,i_4): j_1n\leq  i_1<(j_1+1)n, \  j_2n\leq  i_2<(j_2+1)n, \  j_3n\leq i_3<(j_3+1)n, \ j_4n \leq i_4< (j_4+1)n,\ j_1\neq j_2\neq j_3\neq j_4 \ and\  j_1,j_2,j_3,j_4\in [0,3] \}$  
    with $\alpha_{i_1}+\alpha_{i_3}=\alpha_{i_2}+\alpha_{i_4}$, where the subscripts are reduced modulo $n$.   
 
\end{lemma}

\begin{proof}
    The proof is based on the parity-check matrix of the MP code defined in (\ref{eq3.5}) and Lemma \ref{le2.3}.
    We write the codeword $\textbf{c}$ of $\mathcal{C}$ as $\textbf{c}=(\textbf{c}_1,\textbf{c}_2,\textbf{c}_3,\textbf{c}_4)$, 
    where $\textbf{c}_i$ is a vector of length $n$. 

    (i) If $\textbf{c}_1$, $\textbf{c}_2$, $\textbf{c}_3$ and  $\textbf{c}_4$ are all nonzero vectors, then we can get  $\textbf{c}_i\in \mathcal{C}_1$.
    As $\mathcal{C}_1$ is an $[n,n,1]_q$ code, we have $w_H(\textbf{c})\geq 4$.

    If one of $\textbf{c}_1$, $\textbf{c}_2$, $\textbf{c}_3$, $\textbf{c}_4$ is a zero vector, then we can get  $\textbf{c}_i\in \mathcal{C}_2$. 
    As $\mathcal{C}_2$ is an $[n,n,1]_q$ code, we have $w_H(\textbf{c})\geq 3$. 

    If any two of $\textbf{c}_1$, $\textbf{c}_2$, $\textbf{c}_3$, $\textbf{c}_4$ are zero vectors, then we can get  $\textbf{c}_i\in \mathcal{C}_3$. 
    As $\mathcal{C}_3$ is an $[n,n-1,2]_q$ code, we have $w_H(\textbf{c})\geq 4$. 

    If one of $\textbf{c}_1$, $\textbf{c}_2$, $\textbf{c}_3$, $\textbf{c}_4$ is a nonzero vector, then we can get  $\textbf{c}_i\in \mathcal{C}_4$. 
    As $\mathcal{C}_4$ is an $[n,n-3,4]_q$ code, we have $w_H(\textbf{c})\geq 4$. 
   
    Hence, one can get the Hamming weight of $\textbf{c}$ is always greater than or equal to $3$. Hence, (i) holds.

    (ii) For $w_H(\textbf{c})=3$, only one of $\textbf{c}_1$, $\textbf{c}_2$, $\textbf{c}_3$, $\textbf{c}_4$ is a nonzero vector. 
    Since for any $\{(i_1,i_2,i_3): j_1n\leq  i_1<(j_1+1)n, \  j_2n\leq  i_2<(j_2+1)n, \  j_3n\leq i_3<(j_3+1)n,\ j_1\neq j_2\neq j_3 \ and\  j_1,j_2,j_3\in [0,3] \}$,
    from the form of the parity-check matrix,
    the corresponding column vectors of the parity-check matrix are linearly dependent
    if and only if $\alpha_{i_1}=\alpha_{i_2}=\alpha_{i_3}$, where the subscripts are reduced modulo $n$. 
    Hence, the corresponding codewords exist. (ii) holds.

    (iii) For $w_H(\textbf{c})=4$, any number of $\textbf{c}_1$, $\textbf{c}_2$, $\textbf{c}_3$ and  $\textbf{c}_4$ can be zero  vectors.
    Specially, if $\textbf{c}_1$, $\textbf{c}_2$, $\textbf{c}_3$ and  $\textbf{c}_4$ are all nonzero vectors, 
    then the nonzero vector $\textbf{c}_i$ must satisfy $|supp(\textbf{c}_i)|=1$.
    For $\{(i_1,i_2,i_3,i_4): j_1n\leq  i_1<(j_1+1)n, \  j_2n\leq  i_2<(j_2+1)n, \  j_3n\leq i_3<(j_3+1)n, \ j_4n \leq i_4< (j_4+1)n,\ j_1\neq j_2\neq j_3\neq j_4 \ and\  j_1,j_2,j_3,j_4\in [0,3] \}$, 
    the corresponding column vectors of the parity-check matrix are linearly dependent if and only if 
    with $\alpha_{i_1}+\alpha_{i_3}=\alpha_{i_2}+\alpha_{i_4}$, where the subscripts are reduced modulo $n$.     
    Hence, (iii) holds.

\end{proof}

With the help of the above lemma, the following theorem holds.

\begin{theorem}\label{the3.4}
    Suppose that $q$ is a power of a prime number $p\neq 3$ with $q\equiv 1\ ({\rm mod}\ 4)$, then for each $n\in [4,q]$, there exists an MDS $(4n,6)_q$ symbol-pair code.
\end{theorem}

\begin{proof}
    Since in terms of cosets, $\mathbb{F}_q$ can be written as $\mathbb{F}_q \triangleq \bigcup_{i=0}^{\frac{q}{p}-1}(\chi_i+\mathbb{F}_p)$, where $\chi_0=0$. Let 
    $$\mathcal{F}=(0,1,2,\cdots,p-1,\chi_1,\chi_1+1,\cdots,\chi_1+p-1,\cdots \cdots,\chi_{\frac{q}{p}-1},\cdots,\chi_{\frac{q}{p}-1}+p-1 ),$$
    which consists of all the distinct elements in $\mathbb{F}_q$, and let $\textbf{a}=(\alpha_0,\alpha_1,\cdots,\alpha_{n-1})$ be a vector formed by the first $n$ elements of $\mathcal{F}$.
    
    Let $\mathcal{C}$ be an MP code as defined (\ref{eq3.5}) and $\textbf{c}=(c_0,c_1,\cdots,c_{4n-1})$ be a codeword of $\mathcal{C}$ with length $4n$, whose coordinates is indexed by the set $[0,4n-1]$.
    For each $l\in [0,4n-1]$, we write $l=in+j$, where $i\in\{0,1,2,3\}$, $j\in\{0,1,2,\cdots, n-1\}$.
    Then each entry of the vector $\textbf{c}$ can be represented as $c_{i,j}$. 
    Define a permutation $\psi $ as $\psi (in+j)=i+4j$ and a permutation $\tau $ as 

    \begin{equation}
    \tau (i+4j)=\left\{ 
    \begin{array}{l} \notag 
     i+4(j+2),\ {\rm if}\ i=0,\\
    \\
    i+4(j+1),\ {\rm if}\ i=2,\\
    \\
     i+4j,\ {\rm if}\ i=1,3.\\
    \end{array}
    \right.
    \end{equation}
    
    Namely,

$$c_{0,0},c_{0,1},c_{0,2},c_{0,3},c_{0,4},c_{0,5},\cdots,c_{0,n-1},c_{1,0},\cdots,c_{1,n-1},c_{2,0},\cdots,c_{2,n-1},c_{3,0},\cdots,c_{3,n-1}$$
$$\downarrow\ \psi  $$
$$c_{0,0},c_{1,0},c_{2,0},c_{3,0},c_{0,1},c_{1,1},c_{2,1},c_{3,1},c_{0,2},c_{1,2},c_{2,2},c_{3,2}\cdots,\cdots, c_{0,n-1},c_{1,n-1},c_{2,n-1},c_{3,n-1}$$
$$\downarrow\ \tau $$
$$c_{0,2},c_{1,0},c_{2,1},c_{3,0},c_{0,3},c_{1,1},c_{2,2},c_{3,1},c_{0,4},c_{1,2},c_{2,3},c_{3,2}\cdots,\cdots, c_{0,1},c_{1,n-1},c_{2,0},c_{3,n-1}.$$

    Suppose that $\mathcal{D}$ is a code that permutates equivalent to the code $\mathcal{C}$ under 
    the specific permutations  $\psi$ and $\tau$.
    Namely,
    \begin{equation*}
        \begin{aligned}
            \mathcal{D}: =&\tau(\psi(\mathcal{C}))  \\
                       : =&\{\tau(\psi(\textbf{c})), \  \forall \textbf{c}\in \mathcal{C}\}. \\
        \end{aligned}
    \end{equation*}

    The codes $\mathcal{C}$ and $\mathcal{D}$ have the same parameters $[4n,4n-4,3]_q$ because they are permutation equivalent.
    We will illustrate that $d_{sp}(\mathcal{D})=6$ with the help of the support distribution of the codewords $\textbf{c}$ of $\mathcal{C}$.
   
    From Lemma \ref{le3.4}, there are no codewords with Hamming weight less than $3$. 
    Combining with $w_{sp}(\textbf{c})\geq w_H(\textbf{c})+1$ for $w_H(\textbf{c})<3n$, and $w_{sp}(\textbf{c})=w_H(\textbf{c})$ for $w_H(\textbf{c})=3n$,
    we only need to discuss the cases $3 \leq w_H(\textbf{c})\leq  4$.

    Case I: $w_H(\textbf{c})=3$
    
    According to Lemma \ref{le3.4}, the support of $\textbf{c}$ of weight $3$ must satisfy 
    $\{(i_1,i_2,i_3): j_1n\leq  i_1<(j_1+1)n, \  j_2n\leq  i_2<(j_2+1)n, \  j_3n\leq i_3<(j_3+1)n,\ j_1\neq j_2\neq j_3 \ and\  j_1,j_2,j_3\in [0,3] \}$ 
    with $\alpha_{i_1}=\alpha_{i_2}=\alpha_{i_3}$, where the subscripts are reduced modulo $n$.   
    After the permutations $\psi$ and $\tau$, we can get $w_{sp}(\tau (\psi (\textbf{c})))=6$ 
    except the codewords have the following style:

    Style: $(0,\cdots,0,\star, \star, 0,\cdots,0, \star, 0,\cdots,0)$ or  $(0,\cdots,0,\star, \star, \star, 0,\cdots,0)$

    Since $\textbf{a}$ is a vector which is composed of the first $n$ elements of $\mathcal{F}$ with different elements,
    then the different entries $c_{i,j}$ of the codeword correspond to different $\alpha_j$.
    Hence, $\alpha_j$ is equal to $\alpha_{j^\prime}$ if and only if the second subscript $j$ of $c_{i,j}$ is equal to the second subscript $j^\prime$ of $c_{i^\prime,j^\prime}$.
    After the permutations,
    it is easy to see that there are no consecutive nonzero entries with the same second subscripts but different first subscripts such that the above cases hold.

    Case II: $w_H(\textbf{c})=4$

    If the support of $\textbf{c}$ satisfy
    $\{(i_1,i_2,i_3,i_4):  jn\leq i_1<i_2<i_3<i_4<(j+1)n,\ j\in [0,3] \} $ or
    $\{(i_1,i_2,i_3,i_4): j_1n\leq  i_1<i_2<(j_1+1)n, \ j_2n \leq i_3<i_4< (j_2+1)n,\  j_1\neq j_2\ and\ j_1,j_2 \in [0,3] \}$,
    after the permutations $\psi$ and $\tau$, 
    we can easily get $I\geq 2$ and $w_{sp}(\tau (\psi (\textbf{c}))) \geq 6$.

    If the support of $\textbf{c}$ satisfy 
    $\{(i_1,i_2,i_3,i_4): j_1n\leq  i_1<(j_1+1)n, \  j_2n\leq  i_2<(j_2+1)n, \  j_3n\leq i_3<(j_3+1)n, \ j_4n \leq i_4< (j_4+1)n,\ j_1\neq j_2\neq j_3\neq j_4 \ and\  j_1,j_2,j_3,j_4\in [0,3] \}$  
    with $\alpha_{i_1}+\alpha_{i_3}=\alpha_{i_2}+\alpha_{i_4}$, where the subscripts are reduced modulo $n$,   
    then $w_{sp}(\tau (\psi (\textbf{c})))\geq 6$ except for the cases where the codeword must have four consecutive non-zero entries:
    
    Case II-1: $c_{0,m+2},c_{1,m},c_{2,m+1},c_{3,m}$, where $m\in [0,n-1]$. 

    Case II-2: $c_{1,m},c_{2,m+1},c_{3,m},c_{0,m+3}$, where $m\in [0,n-1]$. 

    Case II-3: $c_{2,m+1},c_{3,m},c_{0,m+3},c_{1,m+1}$, where $m\in [0,n-1]$. 

    Case II-4: $c_{3,m},c_{0,m+3},c_{1,m+1},c_{2,m+2}$, where $m\in [0,n-1]$. 

    Here we only specify that Case II-1 is contradictory to $\alpha_{i_1}+\alpha_{i_3}=\alpha_{i_2}+\alpha_{i_4}$ ($\alpha_{m}+\alpha_{m}=\alpha_{m+1}+\alpha_{m+2}$ for  Case II-1) and Case II-2 to Case II-4 are similar.
    
    Let the subsets $\mathbf{U}_1$ and $\mathbf{U}_2$ of $\mathcal{F}$ be
    $$\mathbf{U}_1=\{p-2, \chi_1+p-2, \chi_2+p-2, \cdots, \chi_{\frac{q}{p}-1}+p-2\},$$
    and 
    $$\mathbf{U}_2=\{p-1, \chi_1+p-1, \chi_2+p-1, \cdots, \chi_{\frac{q}{p}-1}+p-1\}.$$

    If $\alpha_{m}$ belongs to $ \mathbf{U}_1$,
    then it may be assumed that  $\alpha_{m}=\chi_{i_1}+p-2$, where $i_1\in \{0,1,2,\cdots, \frac{q}{p}-1 \}$. 
    Then $\alpha_{m+1}=\chi_{i_1}+p-1$ and  $\alpha_{m+2}=\chi_{i_1+1}$.
    Since $\alpha_{m}+\alpha_{m}=\alpha_{m+1}+\alpha_{m+2}$, 
    we can get $\chi_{i_1}=\chi_{i_1+1}+3$, which clearly does not hold due to $p\neq 3$.

    If $\alpha_{m}$ belongs to $ \mathbf{U}_2$,
    then it may be assumed that  $\alpha_{m}=\chi_{i_1}+p-1$, where $i_1\in \{0,1,2,\cdots, \frac{q}{p}-1 \}$. 
    Then $\alpha_{m+1}=\chi_{i_1+1}$ and  $\alpha_{m+2}=\chi_{i_1+1}+1$.
    Since $\alpha_{m}+\alpha_{m}=\alpha_{m+1}+\alpha_{m+2}$, we can get $2\chi_{i_1}=2\chi_{i_1+1}+3$, which clearly does not hold due to $p\neq 3$.

    Otherwise, if $\alpha_{m}\notin (\mathbf{U}_1\cup \mathbf{U}_2)$, then $\alpha_{m+1}=\alpha_{m}+1$ and $\alpha_{m+2}=\alpha_{m}+2$.
    So $\alpha_{m+1}+\alpha_{m+2}=\alpha_{m}+\alpha_{m}+3$, which contradicts $\alpha_{m+1}+\alpha_{m+2}=\alpha_{m}+\alpha_{m}$ since $p\neq 3$.

    Therefore, by discussing the codewords satisfying $3 \leq w_H(\textbf{c})\leq  4$, we can get $w_{sp}(\mathcal{D})=w_{sp}(\mathcal{\tau (\psi (\textbf{c}))})\geq 6$.
    According to Lemma \ref{le1.1}, $d_{sp}(\mathcal{D})\leq 4n-(4n-4)+2=6$. 
    So $d_{sp}(\mathcal{D})=6$ and $\mathcal{D}$ is an MDS $(4n,6)_q$ symbol-pair code.

\end{proof}

\subsubsection{Symbol-pair distance $d_{sp}(\mathcal{C})=7$ }

Assume that $\mathcal{C}_1$ is the code with parameters $[n,n,1]_q$, $\mathcal{C}_2=\mathcal{C}_3$ is the GRS code $GRS_{1}$ with parameters $[n,n-1,2]_q$,
and $\mathcal{C}_4$ is the GRS code $GRS_{3}$ with parameters $[n,n-4,5]_q$.

Define the MP code
\begin{align}\label{eq3.6}
    \mathcal{C}=[\mathcal{C}_1,\mathcal{C}_2,\mathcal{C}_3,\mathcal{C}_4]\cdot A.
\end{align}
Obviously, $\mathcal{C}$ is a $[4n,4n-6]_q$ code, and from Corollary \ref{co3.1}, the parity-check matrix for $\mathcal{C}$ is shown below:

 \[\begin{split} H &=\frac{1}{4}\begin{pmatrix}
    H_1    & \omega^3H_1  & \omega^2H_1   & \omega H_1 \\
    H_1   & \omega^2 H_1 & H_1 & \omega^2 H_1\\
    H_4   & \omega H_4  & \omega^2H_4 & \omega^3 H_4\\
\end{pmatrix},\\
\setcounter{MaxMatrixCols}{14}
&=\begin{pmatrix}
    1             & \cdots       & 1               & \omega^3           & \cdots     & \omega^3             & \omega^2            & \cdots     & \omega^2 & \omega  & \cdots  & \omega  \\
    1             & \cdots       & 1               &\omega^2           & \cdots      & \omega^2             & 1                   & \cdots    & 1         & \omega^2   & \cdots  & \omega^2   \\
    1             & \cdots       & 1               &    \omega          & \cdots     & \omega               &    \omega^2         & \cdots    & \omega^2     &    \omega^3        & \cdots         & \omega^3     \\                                             
    \alpha_0      & \cdots       & \alpha_{n-1}    &  \omega\alpha_0   & \cdots      & \omega\alpha_{n-1}   &  \omega^2\alpha_0    & \cdots & \omega^2\alpha_{n-1} &    \omega^3\alpha_0  & \cdots    & \omega^3\alpha_{n-1}  \\
    \alpha^2_0    & \cdots       & \alpha^2_{n-1}  &   \omega\alpha^2_0& \cdots      & \omega\alpha^2_{n-1} &   \omega^2\alpha^2_0 & \cdots   &  \omega^2\alpha^2_{n-1} &    \omega^3\alpha^2_0    & \cdots  & \omega^3\alpha^2_{n-1}  \\
    \alpha^3_0    & \cdots       & \alpha^3_{n-1}  &   \omega\alpha^3_0& \cdots      & \omega\alpha^3_{n-1} &   \omega^2\alpha^3_0 & \cdots   &  \omega^2\alpha^3_{n-1} &    \omega^3\alpha^3_0    & \cdots  & \omega^3\alpha^3_{n-1}  \\
\end{pmatrix}.
\end{split} \]

The following lemma determines the support for codewords in $\mathcal{C}$ whose Hamming weights do not exceed $5$.

\begin{lemma}\label{le3.5}
    Let $\mathcal{C}$ be the MP code as defined in (\ref{eq3.6}), and $\textbf{c}$ be a codeword of $\mathcal{C}$ with coordinates indexed by the set $[0,4n-1]$.
    Then the following results hold.
    
    (i) There are no codewords in code $\mathcal{C}$ with Hamming weight less than $4$.

    (ii)  If $w_H(\textbf{c})=4$, then the support  of $\textbf{c}$ must satisfy 
    $\{(i_1,i_2,i_3,i_4): j_1n\leq  i_1<i_2<(j_1+1)n, \  j_2n\leq  i_3<i_4<(j_2+1)n,\  j_1,j_2\in [0,3]  \}$
    with $\alpha_{i_1}=\alpha_{i_3}$ and $ \alpha_{i_2}=\alpha_{i_4}$, where the subscripts are reduced modulo $n$.   
    or
    $\{(i_1,i_2,i_3,i_4): j_1n\leq  i_1<(j_1+1)n, \  j_2n\leq  i_2<(j_2+1)n, \  j_3n\leq i_3<(j_3+1)n, \ j_4n \leq i_4< (j_4+1)n,\ j_1\neq j_2\neq j_3\neq j_4 \ and\  j_1,j_2,j_3,j_4\in [0,3] \}$ 
    with $\alpha_{i_1}=\alpha_{i_3}$ and $ \alpha_{i_2}=\alpha_{i_4}$, where the subscripts are reduced modulo $n$.   

    (iii) If $w_H(\textbf{c})=5$, then the support of $\textbf{c}$ must satisfy 
    $\{(i_1,i_2,i_3,i_4,i_5):  jn\leq i_1<i_2<i_3<i_4<i_5<(j+1)n,\ j\in [0,3]  \} $ or 
    $\{(i_1,i_2,i_3,i_4,i_5): j_1n\leq  i_1<i_2<i_3<(j_1+1)n, \  j_2n\leq  i_4<i_5<(j_2+1)n,\  j_1,j_2\in [0,3]  \}$
    or $\{(i_1,i_2,i_3,i_4,i_5): j_1n\leq  i_1<i_2<(j_1+1)n, \  j_2n\leq  i_3<(j_2+1)n, \  j_3n\leq i_4<(j_3+1)n, \ j_4n \leq i_5< (j_4+1)n,\ j_1\neq j_2\neq j_3\neq j_4 \ and\  j_1,j_2,j_3,j_4\in [0,3]  \}$ 
    with $\alpha_{i_3}= \alpha_{i_1}$ or  $\alpha_{i_2}$,
    $\alpha_{i_4}= \alpha_{i_1}$ or  $\alpha_{i_2}$, and
    $\alpha_{i_5}= \alpha_{i_1}$ or  $\alpha_{i_2}$,
    where the subscripts are reduced modulo $n$.

\end{lemma}

\begin{proof}
    The proof is similar to Lemma \ref{le3.4}, we omit it here.
\end{proof}

With the help of the above lemma, the following theorem holds.
        
\begin{theorem}\label{the3.5}
    Suppose that $q$ is a power of a prime number $p$ with $q\equiv 1\ ({\rm mod}\ 4)$,  then for each $n\in [5,q]$, there exists an AMDS $(4n,7)_q$ symbol-pair code.
\end{theorem}

\begin{proof}
    Since in terms of cosets, we can write $\mathbb{F}_q \triangleq \bigcup_{i=0}^{\frac{q}{p}-1}(\chi_i+\mathbb{F}_p)$, where $\chi_0=0$. Let 
    $$\mathcal{F}=(0,1,2,\cdots,p-1,\chi_1,\chi_1+1,\cdots,\chi_1+p-1,\cdots \cdots,\chi_{\frac{q}{p}-1},\cdots,\chi_{\frac{q}{p}-1}+p-1 ),$$
    which consists of all the distinct elements in $\mathbb{F}_q$, and let $\textbf{a}$ be a vector formed by the first $n$ elements of $\mathcal{F}$.
 
    Let $\mathcal{C}$ be an MP code as defined (\ref{eq3.6}) and $\textbf{c}=(c_0,c_1,\cdots,c_{4n-1})$ be a codeword of $\mathcal{C}$ with length $4n$, whose coordinates is indexed by the set $[0,4n-1]$.
    For each $l\in [0,4n-1]$, we write $l=in+j$, where $i\in\{0,1,2,3\}$, $j\in\{0,1,2,\cdots, n-1\}$.
    Then each entry of the vector $\textbf{c}$ can be represented as $c_{i,j}$. 
    Define a permutation $\psi $ as $\psi (in+j)=i+4j$ and a permutation $\tau $ as 

    \begin{equation}
        \tau (i+4j)=\left\{ 
        \begin{array}{l} \notag 
         i+4j,\ {\rm if}\ i=0,3,\\
        \\
        (i+1)+4(j+1),\ {\rm if}\ i=1,\\
        \\
        (i-1)+4(j-1),\ {\rm if}\ i=2.\\
        \end{array}
        \right.
    \end{equation}

    Namely,
    $$c_{0,0},c_{0,1},c_{0,2},c_{0,3},c_{0,4},c_{0,5},\cdots,c_{0,n-1},c_{1,0},\cdots,c_{1,n-1},c_{2,0},\cdots,c_{2,n-1},c_{3,0},\cdots,c_{3,n-1}$$
    $$\downarrow\ \psi  $$
    $$c_{0,0},c_{1,0},c_{2,0},c_{3,0},c_{0,1},c_{1,1},c_{2,1},c_{3,1},c_{0,2},c_{1,2},c_{2,2},c_{3,2}\cdots,\cdots, c_{0,n-1},c_{1,n-1},c_{2,n-1},c_{3,n-1}$$
    $$\downarrow\ \tau $$
    $$c_{0,0},c_{2,n-1},c_{1,1},c_{3,0}, c_{0,1},c_{2,0},c_{1,2},c_{3,1},c_{0,2},c_{2,1},c_{1,3},c_{3,2}\cdots,\cdots, c_{0,n-1},c_{2,n-2},c_{1,0},c_{3,n-1}$$

    Suppose that $\mathcal{D}$ is a code that permutates equivalent to the code $\mathcal{C}$ under 
    the specific permutations  $\psi$ and $\tau$.
    Namely,
    \begin{equation*}
        \begin{aligned}
            \mathcal{D}: =&\tau(\psi(\mathcal{C}))  \\
                       : =&\{\tau(\psi(\textbf{c})), \  \forall \textbf{c}\in \mathcal{C}\}. \\
        \end{aligned}
    \end{equation*}

    The codes $\mathcal{C}$ and $\mathcal{D}$ have the same parameters $[4n,4n-6,4]_q$ because they are permutation equivalent.
    From Lemma \ref{le3.5}, there are no codewords with Hamming weight less than $4$. 
    Combining with $w_{sp}(\textbf{c})\geq w_H(\textbf{c})+1$ for $w_H(\textbf{c})<4n$, and $w_{sp}(\textbf{c})=w_H(\textbf{c})$ for $w_H(\textbf{c})=4n$,
    We only need to discuss $4\leq w_H(\textbf{c})\leq 5$ in order to show that $d_{sp}(\mathcal{D})=7$.

    Case I: $w_H(\textbf{c})=4$

     If the support of $\textbf{c}$ satisfy 
     $\{(i_1,i_2,i_3,i_4): j_1n\leq  i_1<i_2<(j_1+1)n, \  j_2n\leq  i_3<i_4<(j_2+1)n,\  j_1,j_2\in [0,3]  \}$
     with $\alpha_{i_1}=\alpha_{i_3}$ and $ \alpha_{i_2}=\alpha_{i_4}$, where the subscripts are reduced modulo $n$,
     then $w_{sp}(\mathcal{\tau (\psi (\textbf{c}))})\geq 7$ except the codewords have the following style:

     Style: $(0,\cdots,0,\star, \star, 0,\cdots,0, \star,\star, 0,\cdots,0)$.

     Obviously, after the permutations,
     it is easy to see that there are no consecutive nonzero entries of $c_{i,j}$ with the same second subscripts,
     which contradicts $\alpha_{i_1}=\alpha_{i_3}$ or $ \alpha_{i_2}=\alpha_{i_4}$.
 
     If the support of $\textbf{c}$ satisfy 
     $\{(i_1,i_2,i_3,i_4): j_1n\leq  i_1<(j_1+1)n, \  j_2n\leq  i_2<(j_2+1)n, \  j_3n\leq i_3<(j_3+1)n, \ j_4n \leq i_4< (j_4+1)n,\ j_1\neq j_2\neq j_3\neq j_4 \ and\  j_1,j_2,j_3,j_4\in [0,3] \}$ 
     with $\alpha_{i_1}=\alpha_{i_3}$ and $ \alpha_{i_2}=\alpha_{i_4}$, where the subscripts are reduced modulo $n$.   
     Then $w_{sp}(\mathcal{\tau (\psi (\textbf{c}))})\geq 7$ except the codewords have the following style:

     Style: $(0,\cdots,0,\star, \star, 0,\cdots,0, \star,\star, 0,\cdots,0)$ or $(0,\cdots,0,\star, \star, \star,0,\cdots,0, \star, 0,\cdots,0)$
     
     Similarly, after the permutations, there are no consecutive nonzero entries of $c_{i,j}$ with the same second subscripts,
     which contradicts $\alpha_{i_1}=\alpha_{i_3}$ or $ \alpha_{i_2}=\alpha_{i_4}$.

    Case II: $w_H(\textbf{c})=5$

    If the support of $\textbf{c}$ satisfy
    $\{(i_1,i_2,i_3,i_4,i_5):  jn\leq i_1<i_2<i_3<i_4<i_5<(j+1)n,\ j\in [0,3]  \} $ or 
    $\{(i_1,i_2,i_3,i_4,i_5): j_1n\leq  i_1<i_2<i_3<(j_1+1)n, \  j_2n\leq  i_4<i_5<(j_2+1)n,\  j_1,j_2\in [0,3]  \}$,
    after the permutations $\psi$ and $\tau$, 
    we can easily get $I\geq 3$ and $w_{sp}(\tau (\psi (\textbf{c}))) \geq 8$.

    If the support of $\textbf{c}$ satisfy 
    $\{(i_1,i_2,i_3,i_4,i_5): j_1n\leq  i_1<i_2<(j_1+1)n, \  j_2n\leq  i_3<(j_2+1)n, \  j_3n\leq i_4<(j_3+1)n, \ j_4n \leq i_5< (j_4+1)n,\ j_1\neq j_2\neq j_3\neq j_4 \ and\  j_1,j_2,j_3,j_4\in [0,3]  \}$ 
    with   $\alpha_{i_3}= \alpha_{i_1}$ or  $\alpha_{i_2}$,
    $\alpha_{i_4}= \alpha_{i_1}$ or  $\alpha_{i_2}$, and
    $\alpha_{i_5}= \alpha_{i_1}$ or  $\alpha_{i_2}$,
    where the subscripts are reduced modulo $n$, 
    then $w_{sp}(\tau (\psi (\textbf{c})))\geq 7$ except for the cases where the codeword must have five consecutive non-zero entries:
  
    Case II-1: $c_{0,m},c_{2,m-1},c_{1,m+1},c_{3,m},c_{0,m+1}$, where $m\in [0,n-2]$. 

    Case II-2: $c_{2,m-1},c_{1,m+1},c_{3,m},c_{0,m+1},c_{2,m}$, where $m\in [0,n-2]$. 

    Case II-3: $c_{1,m+1},c_{3,m},c_{0,m+3},c_{2,m+1},c_{1,m+2}$, where $m\in [0,n-2]$. 

    Case II-4: $c_{3,m},c_{0,m+1},c_{2,m},c_{1,m+2},c_{3,m+1},$, where $m\in [0,n-2]$.

   Since $\textbf{a}$ is a vector which consists of the first $n$ elements of  $\mathcal{F}$ with different elements,
   then we can get $\alpha_{m-1}\neq \alpha_{m}\neq \alpha_{m+1}$.
   This is a contradiction due to 
   $\alpha_{i_3}= \alpha_{i_1}$ or  $\alpha_{i_2}$,
   $\alpha_{i_4}= \alpha_{i_1}$ or  $\alpha_{i_2}$, and
   $\alpha_{i_5}= \alpha_{i_1}$ or  $\alpha_{i_2}$,
   (for example, $\alpha_{m-1}\neq \alpha_{m}$ and  $\alpha_{m-1}\neq \alpha_{m+1}$ for Case II-1)

    Therefore,  by discussing the codewords satisfying $3 \leq w_H(\textbf{c})\leq  5$, we can get $w_{sp}(\mathcal{D})=w_{sp}(\tau (\psi (\textbf{c}))) \geq 7$ and $\mathcal{D}$ is an AMDS $(4n,7)_q$ symbol-pair code.

\end{proof}

In the following, for Theorems 3.4-3.5, we will respectively give an example.
    
\begin{example}\label{example3.4}
    Let $q=p=5$, and $\mathbb{F}_5=\{0,1,2,3,4\}$.
    Taking $\omega=2$ be a primitive $4$-th root of unity and 
    $$ A=\begin{pmatrix}
        1   & 1   & 1 & 1 \\
        1   & 2   & 4 & 3   \\
        1   & 4   & 1 & 4\\
        1   & 3   & 4 & 2\\

        \end{pmatrix}.$$

Then we can get 
    $$ (A^{-1})^T=\frac{1}{4}\begin{pmatrix}
        1   & 1   & 1 & 1 \\
        1   & 3   & 4 & 2  \\
        1   & 4   & 1 & 4 \\
        1   & 2   &4  & 3   \\
    \end{pmatrix}.$$

    Assume that $\mathcal{C}_1=\mathcal{C}_2$ is the code with parameters $[5,5,1]_5$, 
$\mathcal{C}_3$ is the GRS code with parameters $[5,4,2]_5$ whose parity-check matrix is 
$$ H_1=\begin{pmatrix}
    1   & 1   & 1 &  1 &  1  \\
\end{pmatrix},$$

$\mathcal{C}_3$ is the GRS code with parameters $[5,2,4]_5$ whose parity-check matrix is 
$$ H_2=\begin{pmatrix}
    1   & 1   & 1 &  1  &  1  \\
    0   & 1   & 2 &  3  &  4 \\
    0   & 1   & 4 &  4  &  1 \\
\end{pmatrix}.$$

The parity-check matrix of the MP code $\mathcal{C}=[\mathcal{C}_1,\mathcal{C}_2,\mathcal{C}_3,\mathcal{C}_4]\cdot A$ with parameters $[20,16,3]_5$ is 
\setcounter{MaxMatrixCols}{50}
$$ H=\frac{1}{3}\begin{pmatrix}
    1   & 1   & 1 &  1 & 1     &    4 & 4 & 4 & 4 & 4     &    1 & 1 & 1 & 1 & 1   &    4 & 4 & 4 & 4 & 4  \\
    1   & 1   & 1 &  1 & 1     &    2 & 2 & 2 & 2 & 2     &    4 & 4 & 4 & 4 & 4   &    3 & 3 & 3 & 3 & 3\\
    0   & 1   & 2 &  3 & 4     &    0 & 2 & 4 & 1 & 3     &    0 & 4 & 3 & 2 & 1   &    0 & 3 & 1 & 4 & 2\\
    0   & 1   & 4 &  4 & 1     &    0 & 2 & 3 & 3 & 2     &    0 & 4 & 1 & 1 & 4   &    0 & 3 & 2 & 2 & 3\\

\end{pmatrix}.$$

After the following specific permutations, 

$$1,2,3,4,5,6,7,8,9,10,11,12,13,14,15,16,17,18,19,20$$
$$\downarrow\ \psi$$
$$1,6,11,16, 2,7,12,17, 3,8,13,18, 4,9,14,19, 5,10,15,20$$
$$\downarrow\ \tau$$
$$3,6,12,16, 4,7,13,17, 5,8,14,18, 1,9,15,19, 2,10,11,20,$$

the code $\mathcal{D}$ with parameters $[20,16,3]_5$ has a parity-check matrix of the form:

$$ H=\frac{1}{3}\begin{pmatrix}
    1   & 4   & 1 &  4 & 1     &    4 & 1 & 4 & 1 & 4     &    1 & 4 & 1 & 4 & 1   &    4 & 1 & 4 & 1 & 4  \\
    1   & 2   & 4 &  3 & 1     &    2 & 4 & 3 & 1 & 2     &    4 & 3 & 1 & 2 & 4   &    3 & 1 & 2 & 4 & 3\\
    2   & 0   & 4 &  0 & 3     &    2 & 3 & 3 & 4 & 4     &    2 & 1 & 0 & 1 & 1   &    4 & 1 & 3 & 0 & 2\\
    4   & 0   & 4 &  0 & 4     &    2 & 1 & 3 & 1 & 3     &    1 & 2 & 0 & 3 & 4   &    2 & 1 & 2 & 0 & 3\\

\end{pmatrix}.$$

According to Theorem \ref{the3.4} and the Magma program, the code $\mathcal{D}$ is an MDS $(20,6)_5$ symbol-pair code.

\end{example}

\begin{example}\label{example3.5}
    Let $p=3$, $q=9$, and $\mathbb{F}_9=\{0,1,2,\alpha,\alpha+1,\alpha+2,2\alpha,2\alpha+1,2\alpha+2\}$, where $\alpha$ is a root of the irreducible polynomial $f(x)=x^2+1$ over $\mathbb{F}_3$.
    Let $\xi=1+\alpha$ be  a primitive element of $\mathbb{F}_9$, then $\xi^2=2\alpha$, $\xi^3=1+2\alpha$, $\xi^4=2$, $\xi^5=2+2\alpha$, $\xi^6=\alpha$, $\xi^7=2+\alpha$, $\xi^8=1$.
    Taking $\omega=\xi^2$ be a primitive $4$-th root of unity and 
    $$ A=\begin{pmatrix}
        1   & 1   & 1  & 1  \\
        1   & \xi^2  & \xi^4  & \xi^6   \\
        1   & \xi^4   & 1 & \xi^4   \\
        1   & \xi^6    & \xi^4   & \xi^2     \\
    \end{pmatrix}.$$
Then we can get 
    $$ (A^{-1})^T=\frac{1}{4}\begin{pmatrix}
        1   & 1   & 1 & 1 \\
        1   &\xi^6   & \xi^4    & \xi^2\\
        1   & \xi^4   &1 & \xi^4   \\
        1   & \xi^2  & \xi^4  & \xi^6  \\
    \end{pmatrix}.$$ 
    
    Assume that $\mathcal{C}_1$ is the code with parameters $[6,6,1]_9$,
    $\mathcal{C}_2=\mathcal{C}_3$ is the GRS code with parameters $[6,5,2]_9$ whose parity-check matrix is 
    $$ H_2=\begin{pmatrix}
        1   & 1   & 1 &  1 & 1 & 1 \\
    \end{pmatrix},$$
    $\mathcal{C}_4$ is the GRS code with parameters $[6,2,5]_9$ whose parity-check matrix is 

    $$ H_3=\begin{pmatrix}
        1   & 1   & 1     &  1 & 1 & 1 \\
        0   & 1   & \xi^4 &  \xi^6 & \xi & \xi^7 \\
        0   & 1   & 1     &  \xi^4 & \xi^2 & \xi^6 \\
        0   & 1   & \xi^4 &   \xi^2& \xi^3 & \xi^5
    \end{pmatrix}.$$

    The parity-check matrix of the MP code $\mathcal{C}=[\mathcal{C}_1,\mathcal{C}_2,\mathcal{C}_3,\mathcal{C}_4]\cdot A$ with parameters $[24,18,4]_9$ is 

    \setcounter{MaxMatrixCols}{30}
    $$ H=\frac{1}{4}\begin{pmatrix}
        1   & 1   & 1     &  1     & 1     & 1     & \xi^6 & \xi^6  & \xi^6  & \xi^6 & \xi^6  & \xi^6  &  \xi^4 & \xi^4  & \xi^4 & \xi^4 & \xi^4  & \xi^4 & \xi^2& \xi^2 & \xi^2 & \xi^2 & \xi^2 & \xi^2 \\
        1   & 1   & 1     &  1     & 1     & 1     & \xi^4 & \xi^4  & \xi^4  & \xi^4 & \xi^4  & \xi^4  &  1     & 1      & 1     &  1    & 1      & 1     &\xi^4 & \xi^4 & \xi^4 & \xi^4 & \xi^4 & \xi^4   \\
        1   & 1   & 1     &  1     & 1     & 1     & \xi^2 &  \xi^2 &  \xi^2 &  \xi^2&  \xi^2 &  \xi^2 & \xi^4  & \xi^4  & \xi^4 & \xi^4 & \xi^4  & \xi^4 &\xi^6 & \xi^6 & \xi^6 & \xi^6 & \xi^6 & \xi^6   \\
        0   & 1   & \xi^4 &  \xi^6 & \xi   & \xi^7 & 0     &  \xi^2 & \xi^6  &  1    &  \xi^3 &  \xi   & 0      &  \xi^4 &    1  & \xi^2 &  \xi^5 &\xi^3  &  0   & \xi^6 & \xi^2 & \xi^4 & \xi^7 & \xi^5    \\
        0   & 1   & 1     &  \xi^4 & \xi^2 & \xi^6 & 0     &  \xi^2 &  \xi^2 & \xi^6 & \xi^4  &  1     &  0     &   \xi^4& \xi^4 &   1   & \xi^6  &\xi^2  &   0  & \xi^6 & \xi^6 & \xi^2 &   1   & \xi^4     \\
        0   & 1   & \xi^4 &   \xi^2& \xi^3 & \xi^5 & 0     &  \xi^2 & \xi^6  & \xi^4 &  \xi^5 & \xi^7  &  0     &  \xi^4 &  1    & \xi^6 & \xi^7  &\xi    &   0  &  \xi^6& \xi^2 & 1     &   \xi & \xi^3   \\
    \end{pmatrix}.$$

    After the following specific permutations, 

    $$1,2,3,4,5,6,7,8,9,10,11,12,13,14,15,16,17,18,19,20,21,22,23,24$$
    $$\downarrow\ \psi$$
    $$1,7,13,19,2,8,14,20,3,9,15,21,4,10,16,22,5,11,17,23,6,12,18,24$$
    $$\downarrow\ \tau$$
    $$1,18,8,19,2,13,9,20,3,14,10,21,4,15,11,22,5,16,12,23,6,17,7,24,$$

    the code $\mathcal{D}$ with parameters $[24,18,4]_9$ has a parity-check matrix of the form:

    \setcounter{MaxMatrixCols}{30}
    $$ H=\frac{1}{4}\begin{pmatrix}
        1   & \xi^4   & \xi^6 &  \xi^2 & 1     & \xi^4 & \xi^6 & \xi^2  & 1      & \xi^4 & \xi^6  & \xi^2  &  1     & \xi^4  & \xi^6 & \xi^2 & 1      & \xi^4 & \xi^6& \xi^2 & 1     & \xi^4 & \xi^6 & \xi^2 \\
        1   & 1       & \xi^4 &  \xi^4 & 1     & 1     & \xi^4 & \xi^4  & 1      & 1     & \xi^4  & \xi^4  &  1     & 1      & \xi^4 &  \xi^4& 1      & 1     &\xi^4 & \xi^4 & 1     & 1     & \xi^4 & \xi^4   \\
        1   & \xi^4   & \xi^2 &  \xi^6 & 1     & \xi^4 & \xi^2 &  \xi^6 &  1     &  \xi^4&  \xi^2 &  \xi^6 & 1      & \xi^4  & \xi^2 & \xi^6 & 1      & \xi^4 &\xi^2 & \xi^6 & 1     & \xi^4 & \xi^2 & \xi^6   \\
        0   & \xi^3   & \xi^2 &  0     & 1     & 0     & \xi^6 &  \xi^6 & \xi^4  &  \xi^4&  1     &  \xi^2 & \xi^6  &  1     &  \xi^3& \xi^4 &  \xi   &\xi^2  &  \xi & \xi^7 & \xi^7 & \xi^5 & 0     & \xi^5    \\
        0   & \xi^2   &  \xi^2&  0     & 1     & 0     & \xi^2 &  \xi^6 &  1     & \xi^4 & \xi^6  &  \xi^6 &  \xi^4 &   \xi^4& \xi^4 &  \xi^2& \xi^2  &1      &   1  & 1     & \xi^6 & \xi^6 &   0   & \xi^4     \\
        0   & \xi     & \xi^2 &  0     & 1     & 0     & \xi^6 &  \xi^6 & \xi^4  & \xi^4 &  \xi^4 & \xi^2  &  \xi^2 &  1     &  \xi^5& 1     & \xi^3  &\xi^6  & \xi^7&  \xi& \xi^5 & \xi^7 &   0   & \xi^3   \\
    \end{pmatrix}.$$

    According to Theorem \ref{the3.5} and the Magma program, the code $\mathcal{D}$ is an AMDS $(24,7)_9$ symbol-pair code.

\end{example}

\section{Conclusion}\label{sec4}
In this paper, inspired by the idea in \cite{LELP23}, several new classes of symbol-pair codes are derived from the permutation equivalence codes of 
matrix-product codes. 
Our results extended some conclusions in \cite{LELP23}, which made the lengths of MDS symbol-pair codes more general.
Notice that most of the known MDS symbol-pair codes over  $\mathbb{F}_{q}$, where $q$ is a prime power, have minimum symbol-pair distances $d_{sp}(\mathcal{C})\leq 6$.
If one restricts the finite field to the field with a prime number elements, i.e., $q=p$, then
there exist some constructions of MDS symbol-pair codes with $d_{sp}(\mathcal{C})>6$ under some constraints.
However, MDS symbol-pair codes constructed from matrix-product codes can break such restrictions. 
In our constructions, the MDS symbol-pair codes are over finite field $\mathbb{F}_{q}$ with prime power elements and the minimum symbol-pair distances  $d_{sp}(\mathcal{C})> 6$.
The research in this paper further shows that matrix-product codes is a good source in constructing symbol-pair codes.
We would like to try to use other matrix-product codes to construct symbol-pair codes and derive more new MDS symbol-pair codes in the future.

\section*{Conflict of interest}
The authors declare that there is no possible conflict of interest.

\section*{Code availability}
Not applicable.
\section*{Data availability}
Data sharing is not applicable to this article as no datasets were generated or analyzed during the current study. 

\section*{Acknowledgement}
The work was supported by the National Natural Science Foundation of China (12271137, U21A20428, 12171134).

\end{document}